\def\draft{1}  

\documentclass[11pt]{article}

\usepackage[letterpaper,hmargin=1in,vmargin=1in]{geometry}
\usepackage[T1]{fontenc}

\usepackage{color}                
\usepackage{booktabs}
\usepackage{graphicx}
\usepackage{amssymb}
\usepackage{epstopdf}
\usepackage{amsfonts,latexsym,graphics}
\usepackage{epsfig}
\usepackage{amssymb}
\usepackage{amsmath}
\usepackage{amsfonts}
\usepackage{algorithm}
\usepackage{algorithmic}
\usepackage{url}
\usepackage{rotating}
\usepackage[colorlinks=true,linktoc=all]{hyperref}
\usepackage{eepic}
\usepackage{sectsty}
\usepackage{tikz}
\usetikzlibrary{arrows}

\newcommand{\nc}{\newcommand}
\nc{\rnc}{\renewcommand}

\ifnum\draft=1
\newcommand{\Hnote}[1]{\textcolor{blue}{[\bf Aram's Note: #1]}}
\newcommand{\Nnote}[1]{{\textcolor{red}{[\bf Anand's Note: #1]}}}
\newcommand{\Wnote}[1]{{[\bf Xiaodi's Note: #1]}}

\else
\newcommand{\Hnote}[1]{{}}
\newcommand{\Nnote}[1]{{}}
\newcommand{\Wnote}[1]{{}}
 
\fi

\newtheorem{theorem}{Theorem}
\numberwithin{theorem}{section}
\numberwithin{equation}{section}
\newtheorem{definition}[theorem]{Definition}

\newtheorem{lemma}[theorem]{Lemma}

\newtheorem{proposition}[theorem]{Proposition}

\newtheorem{corollary}[theorem]{Corollary}

\newtheorem{remk}[theorem]{Remark}

\newenvironment{proof}{\noindent{\bf Proof. }}{\qed}

\def\FullBox{\hbox{\vrule width 8pt height 8pt depth 0pt}}

\def\qed{\ifmmode\qquad\FullBox\else{\unskip\nobreak\hfil
\penalty50\hskip1em\null\nobreak\hfil\FullBox
\parfillskip=0pt\finalhyphendemerits=0\endgraf}\fi}

\def\qedsketch{\ifmmode\Box\else{\unskip\nobreak\hfil
\penalty50\hskip1em\null\nobreak\hfil$\Box$
\parfillskip=0pt\finalhyphendemerits=0\endgraf}\fi}

\newcommand{\R}{{\mathbb R}}
\newcommand{\N}{{\mathbb{N}}}

\newcommand{\poly}{{\mathrm{poly}}}
\newcommand{\polylog}{{\mathrm{polylog}}}

\newcommand{\zo}{\{0,1\}}

\DeclareMathOperator*{\E}{\mathbb{E}}

\def\01{\{0,1\}}
\def\eps{\epsilon}

\newcommand{\set}[1]{{\left\{#1\right\}}}

\newcommand{\Prob}{{\mathbf{Pr}}}
\newcommand{\tinyspace}{\mspace{1mu}}
\newcommand{\microspace}{\mspace{0.5mu}}

\newcommand{\norm}[1]{\left\lVert\tinyspace#1\tinyspace\right\rVert}

\newcommand{\tr}{\operatorname{tr}}
\newcommand{\rank}{\operatorname{rank}}

\newcommand{\ip}[2]{\left\langle #1 , #2\right\rangle}
\newcommand{\ipp}[1]{\left\langle #1 \right\rangle}

\def\({\left(}
\def\){\right)}
\def\I{\mathsf{Id}}

\newcommand{\setft}[1]{\mathrm{#1}}
\newcommand{\lin}[1]{\setft{L}\left(#1\right)}
\newcommand{\density}[1]{\setft{Dens}\left(#1\right)}

\newcommand{\ot}{\otimes}

\def\complex{\mathbb{C}}

\def\<{\langle}
\def\>{\rangle}
\def \lket {\left|}
\def \rket {\right\rangle}
\def \lbra {\left\langle}
\def \rbra {\right|}
\newcommand{\ket}[1]{\lket\microspace #1 \microspace\rket}
\newcommand{\bra}[1]{\lbra\microspace #1 \microspace\rbra}
\newcommand{\ketbra}[1]{\lket\microspace #1 \rangle \langle #1 \microspace\rbra}

\nc\yes{\text{yes}}
\nc\no{\text{no}}
\nc\onevec{\vec{\mathbf{1}}}

\newcommand{\trnorm}[1]{\norm{#1}_1}

\newcommand{\commentout}[1]{}

\newenvironment{protocol*}[1]
  {
    \begin{center}
      \hrulefill\\
      \textbf{#1}
  }
  {
    \vspace{-1\baselineskip}
    \hrulefill
    \end{center}
  }

\DeclareMathOperator{\Sep}{Sep}
\DeclareMathOperator{\psE}{\tilde{\E}}
\DeclareMathOperator{\Tr}{Tr}
\DeclareMathOperator{\M3XR}{MAX-3XOR}
\DeclareMathOperator{\MXST}{MAX-SAT}
\DeclareMathOperator{\conv}{conv}
\newcommand{\txor}{\textsc{3XOR}}
\newcommand{\tof}{\textsc{2-out-of-4}}
\newcommand{\tofsat}{\textsc{2-out-of-4-SAT}}
\newcommand{\tofsateq}{\textsc{2-out-of-4-SAT-EQ}}
\newcommand{\qmat}{\textsc{QMA(2)-Honest}}
\newcommand{\qmak}{\textsc{QMA(k)-Honest}}

\def \zo {\{\pm 1\}}
\def \OPT {{\mathrm{OPT}}}

\DeclareMathOperator{\degsos}{\deg_{\mathsf{SOS}}}

\nc{\eq}[1]{(\ref{eq:#1})}
\nc{\eqs}[2]{(\ref{eq:#1}) and (\ref{eq:#2})}
\nc{\ra}{\rightarrow}
\nc{\Ra}{\Rightarrow}

\def\be#1\ee{\begin{equation}#1\end{equation}}
\def\bea#1\eea{\begin{eqnarray}#1\end{eqnarray}}
\def\beas#1\eeas{\begin{eqnarray*}#1\end{eqnarray*}}
\def\ba#1\ea{\begin{align}#1\end{align}}
\def\bas#1\eas{\begin{align*}#1\end{align*}}
\def\bpm#1\epm{\begin{pmatrix}#1\end{pmatrix}}
\nc{\non}{\nonumber}

\nc\cA{\mathcal{A}}
\nc\cB{\mathcal{B}}
\nc\cC{\mathcal{C}}
\nc\cD{\mathcal{D}}
\nc\cE{\mathcal{E}}
\nc\cF{\mathcal{F}}
\nc\cG{\mathcal{G}}
\nc\cH{\mathcal{H}}
\nc\cI{\mathcal{I}}
\nc\cJ{\mathcal{J}}
\nc\cK{\mathcal{K}}
\nc\cL{\mathcal{L}}
\nc\cM{\mathcal{M}}
\nc\cN{\mathcal{N}}
\nc\cO{\mathcal{O}}
\nc\cP{\mathcal{P}}
\nc\cQ{\mathcal{Q}}
\nc\cR{\mathcal{R}}
\nc\cS{\mathcal{S}}
\nc\cT{\mathcal{T}}
\nc\cU{\mathcal{U}}
\nc\cV{\mathcal{V}}
\nc\cW{\mathcal{W}}
\nc\cX{\mathcal{X}}
\nc\cY{\mathcal{Y}}
\nc\cZ{\mathcal{Z}}

\DeclareMathOperator{\NP}{\mathsf{NP}}

\DeclareMathOperator{\Sym}{Sym}
\DeclareMathOperator{\QMA}{\mathsf{QMA}}

\nc\benum{\begin{enumerate}}
\nc\eenum{\end{enumerate}}
\nc\bit{\begin{itemize}}
\nc\eit{\end{itemize}}

\newcommand{\secref}[1]{Section~\ref{sec:#1}}

\newcommand{\lemref}[1]{Lemma~\ref{lem:#1}}
\newcommand{\thmref}[1]{Theorem~\ref{thm:#1}}
\newcommand{\propref}[1]{Proposition~\ref{prop:#1}}

\newcommand{\corref}[1]{Corollary~\ref{cor:#1}}

\newcommand{\tabref}[1]{Table~\ref{tab:#1}}

\nc\bbC{\mathbb{C}}
\nc\bbF{\mathbb{F}}
\nc\bbM{\mathbb{M}}
\nc\bbN{\mathbb{N}}
\nc\bbR{\mathbb{R}}
\nc\bbZ{\mathbb{Z}}

\def\begsub#1#2\endsub{\begin{subequations}\label{eq:#1}\begin{align}#2\end{align}\end{subequations}}
\nc\qand{\qquad\text{and}\qquad}
\nc\mnb[1]{\medskip\noindent{\bf #1}}

\nc\SDPEF{SDP extended formulation}

\title{Limitations of semidefinite programs for separable states and entangled games}
\author{%
 Aram W. Harrow\thanks{Center for Theoretical Physics,
Massachusetts Institute of Technology, Cambridge, MA 02139, USA}  
\and
Anand Natarajan\footnotemark[1]
\and
 Xiaodi Wu\thanks{Department of Computer Science, Institute for Advanced
   Computer Studies, and Joint Center for Quantum Information and Computer
   Science, University of Maryland, College Park, MD 20742, USA}
}

\begin{document}

\begin{titlepage}
\maketitle

\begin{abstract}
Semidefinite programs (SDPs) are a framework for exact or approximate
optimization with widespread application in quantum information theory.
We introduce a new method for using reductions to construct integrality gaps for SDPs, meaning instances where the SDP value is far from the true
  optimum.  These are based on new limitations on the sum-of-squares (SoS)
  hierarchy in approximating two particularly important sets in quantum information
  theory, where previously no $\omega(1)$-round integrality gaps were known:
\begin{enumerate}
\item The set of separable (i.e. unentangled) states, or equivalently, the $2 \ra 4$ norm
  of a matrix.
\item The set of quantum correlations; i.e. conditional probability distributions
  achievable with local measurements on a shared entangled state.
\end{enumerate}
Integrality gaps for the $2\ra 4$ norm had previously been sought due to its connection to
Small-Set Expansion (SSE) and Unique Games (UG).

In both cases no-go theorems were previously known based on computational assumptions such
as the Exponential Time Hypothesis (ETH) which asserts that 3-SAT requires exponential
time to solve. Our unconditional results achieve the same parameters as all of these
previous results (for separable states) or as some of the previous results (for quantum
correlations). In some cases we can make use of the framework of Lee-Raghavendra-Steurer
(LRS) to establish integrality gaps for any SDP extended formulation, not only the SoS hierarchy. Our hardness
result on separable states also yields a dimension lower bound of approximate
disentanglers, answering a question of Watrous and Aaronson et al.

These results can be viewed as limitations on the monogamy principle, the PPT test, the
ability of Tsirelson-type bounds to restrict quantum correlations, as well as the SDP
hierarchies of Doherty-Parrilo-Spedalieri, Navascues-Pironio-Acin and
Berta-Fawzi-Scholz. Indeed a wide range of past work in quantum information can be
described as using an SDP on one of the above two problems and our results put broad
limits on these lines of argument.
\end{abstract}

\vfill
\thispagestyle{empty}

\end{titlepage}

\thispagestyle{empty}
\tableofcontents
\thispagestyle{empty}
\newpage
\setcounter{page}{1}
\section{Introduction} \label{sec:intro}

\subsection{2-to-4 Norms and Separability of Quantum States}

\subsubsection*{Separable (unentangled) quantum states}

Entanglement is an essential ingredient in many applications of quantum information
processing. Understanding the characterization of entangled quantum states remains a fundamental problem in quantum information processing research. 
For example, studying the boundary between entangled and separable quantum states has been
useful for a variety of problems in quantum information such as data hiding~\cite{DLT02},
teleportation~\cite{Masanes96}, privacy~\cite{bauml15}, channel capacities~\cite{Rains01,
  MSW, MW12}, and the quantum marginal problem~\cite{marginals}.  

Distinguishing between entangled and separable quantum states, also known as the separability problem,  turns out to be closely related to the following optimization problem $h_{\Sep(d,d)}$, defined for a positive-semidefinite $d^2\times d^2$ matrix $M$ as
\be h_{\Sep(d,d)}(M) := \max_{\substack{x,y\in\mathbb{C}^d \\ \|x\|_2 = \|y\|_2 = 1}}
\sum_{i,j,k,l \in [d]} M_{ij,kl} x_i^* x_j y_k^* y_l.\label{eq:hSep}\ee

In general these problems cannot be approximated in polynomial time even to constant error~\cite{BHKSZ12,HM13}, assuming the ETH.
Let $\Sep^k(d)$ denote the convex hull of $\ketbra{\psi_1} \ot \cdots \ot \ketbra{\psi_k}$ as $\ket{\psi_1},\ldots,\ket{\psi_k}$ range over all unit vectors in $\bbC^d$. ($h_{\Sep(d,d)}$ is $h_{\Sep^2(d)}$.) For a general convex set $S$, let $h_S(x) := \max_{y\in S} \ip{x}{y}$.    Consider the problem of determining whether $h_{\Sep^k(d)}(M)$ is $\geq c$ or $\leq s$ for some $0\leq s < c \leq 1$ and some matrix $M$ such that $0\leq M \leq I$.   
Several hardness results are known of the form ``determining satisfiability of 3-SAT instances with $n$ variables and $O(n)$ clauses can be reduced to estimating $h_{\Sep^k(d)}$ in this way.''   We summarize such hardness results in Table~\ref{tab:ETH}.

\begin{table}[h]
\begin{center}
\begin{tabular}{cccccc}
\toprule
reference & $k$ &  $c$ & $s$ & $n$ & notes \\
\midrule
\cite{GNN12} & 2 & 1 & $1-\frac{1}{d\cdot \poly\log(d)}$ & $O(d)$ & (1)
 \\
\cite{pereszlenyi:2012} & 2 & 1 & $1-\frac{1}{\poly(d)}$ & $O(d)$ & (2) \\
\cite{ABDFS08} & $\sqrt{d}\cdot \poly\log(d)$ & $1$ & 0.99 & $O(d)$ & (3)\\
\cite{CD10} & $\sqrt{d}\cdot \poly\log(d)$ & $1-2^{-d}$ & 0.99 & $O(d)$ & (4) \\
\cite{HM13} & 2 & 1 & 0.01  & $\frac{\log^2(d)}{\poly\log\log(d)}$ & (5) \\
\bottomrule
\end{tabular}
\caption{Hardness results for $h_{\Sep^k(d)}$.  Notes: (1) This builds on work in \cite{gurvits:2003, BT07, Bei10} which achieved the same result with $s=1-1/\poly(d)$. Related results were found for testing membership in $\Sep^2(d)$ in \cite{gurvits:2003, liu:2007,gharibian10}.  (2) The measurement $M$ can be implemented using a uniform quantum circuit of size $\poly\log(d)$.
(3,4,5) Here 0.99 refers to a constant strictly less than 1 whose explicit value is not known, and 0.01 means the result applies for any constant in the range $(0,1)$.  (4) The measurement $M$ can be taken to be a Bell measurement, meaning that all the systems are measured locally and then the answers are processed classically. (5) $M$ can be taken to be separable, i.e.~of the form $M=\sum_i A_i \ot B_i$ for $A_i,B_i\geq 0$.
\label{tab:ETH}}
\end{center}
\end{table}

These results can be thought of as ETH-based no-go results, since in each case ETH implies a lower bound on the run-time of any algorithm approximating $h_{\Sep}$, and in particular implies the existence of integrality gaps for the SoS hierarchy.   We mention also one hardness result that does not fit into this framework is the result by \cite{BHKSZ12} that a constant-factor multiplicative approximation to $h_{\Sep^2(d)}$ could be used to solve Unique Games instances of size $d^{\Omega(1)}$.

\subsubsection*{DPS hierarchy for separability problem and
integrality gaps } Despite the worst-case hardness, a variety of
heuristics have been developed for the separability problem given the
utility of solutions even for specific cases.

 The set of entangled states was first approximated by the set of
states with non-positive partial transpose~\cite{PPT1,PPT2}.  The
resulting test is known as the ``PPT test'' and it is known that all
separable states have positive semidefinite partial transpose
(i.e.~are PPT) and that some entangled states are PPT while others are not.
 
Doherty, Parrilo and Spedalieri improved this to a hierarchy of
approximations~\cite{DohertyPS04}.  The $k^{\text{th}}$ level of the so-called DPS
hierarchy approximates the set of entangled states by the set of states $\rho^{AB}$ for
which there does not exist $\tilde{\rho}^{A_1\ldots A_kB_1\ldots B_k}$ with
$\rho^{AB}=\tilde\rho^{A_1B_1}$, the supports of $\tilde{\rho}^{A_1\ldots A_k}$ and
$\tilde{\rho}^{B_1\ldots B_k}$ contained in the symmetric subpace and $\tilde\rho$
remaining positive semidefinite under the partial transpose of any set of subsystems.
Again all separable states pass the level-$k$ DPS test as do some entangled states.  As
$k\ra\infty$ the run-time increases exponentially but the accuracy also increases, meaning
that fewer entangled states pass the test.  Since the DPS sets include all separable
states as well as some entangled states, we call the approximation a ``relaxation'' of the
set of separable states, and call entangled states which pass the DPS test ``integrality
gaps''.

The accuracy of the DPS hierarchy has been analyzed by a long sequence of works which have
found various positive results, matching the barriers from ETH in a few
cases~\cite{DF80,CFS02,KR05, KM09, CKMR07, BCY}.  A handful of negative results are also
known but generally only for weaker versions of the DPS hierarchy.  If we require only
that $\tilde\rho$ be symmetric (i.e.~commute with permutations), then the antisymmetric
state is a potent integrality gap showing that this weaker hierarchy still makes large
errors until $k> d$.  This can also be turned into a integrality gap for the slightly
stronger hierarchy which restricts the support of $\tilde\rho$ to be contained in the
symmetric subspace, but it is easily detected with the PPT test.  Another integrality gap
is known to defeat the larger class of tests for which $\tilde\rho$ is required to be
symmetric and PPT across any cut, but only works up to $k=O(\log d)$~\cite{BCY}.  Two
recent results regarding the constant-error case are of particular interest are those of
Lancien~\cite{Lancien15} and Aubrun and Szarek~\cite{Sep-polytope}. The result of Lancien
uses a probabilistic argument to show that most $d$-dimensional non-separable states will
pass the $k$-extendibility test together with the requirement that the state be PPT across
a single cut up to $k =O(d)$. This is a better lower bound on the level of relaxation than
our result, but only for a weaker version of the DPS hierarchy. The result of Aubrun and
Szarek studies the set of positive but not completely positive maps; it is known that any
non-separable state can be detected by such a map. Using tools from convex geometry, they
show a lower bound of $\exp(\Omega(d^3/\log d))$ for the number of positive maps needed to
detect all $d$-dimensional states that are at least a constant distance away from
separable. This can be viewed as a lower bound on the dimension of a specific class of
SDPs, which have a constraint of the form
$\mathrm{Id} \ot (\Phi_1(\rho) \oplus \Phi_2(\rho) \oplus \dots \oplus\Phi_N(\rho))
\succeq 0$,
where $\Phi_1 \dots \Phi_N$ are positive maps on $d$-dimensional matrices. This class of
SDPs contains the PPT test as well as many other SDPs, but it does not contain higher
levels of DPS.

All of these integrality gaps except that of \cite{Lancien15} are known to be defeated by the full DPS hierarchy and there is no known way to modify them to avoid this.  Indeed the only previously known unconditional negative result was in the original DPS paper which showed that the error always remained nonzero for all finite values of $k$ (see also \cite{BeigiS10} showing that this could be amplified).  Indeed one can even define an improved version of DPS that removes this limitation and always exactly converges at a finite (but large) value of $k$~\cite{HNW15a}.

We remark that the goal of much of this past work has been much more general than separability testing.  The convergence proofs of DPS are related to ``monogamy relations'' which bound how widely an entangled state can be shared.  These in turn have been related to the security of quantum key distribution, rigorous proofs of the mean-field approximation in many-body physics~\cite{BH-product}, quantum interactive proofs and many other applications.  
In some cases also integrality gaps for weaker hierarchies have been useful examples of extremal information-theoretic behavior~\cite{CSW12}.

\subsubsection*{Semidefinite Programming (SDP) and Sum-of-Squares (SoS) hierarchies}
The DPS hierarchy is a special case of a more general approach to polynomial optimization
problems known as the Sum-of-Squares (SoS) hierarchy.  The SoS hierarchy in turn is an
example of a Semidefinite Programming (SDP) relaxation~\cite{boyd_book}.  A major question
in the theory of algorithms and complexity is the power of SDP relaxations and in
particular the SoS hierarchy~\cite{BS14}.  Most problems in NP admit various SDP
relaxations, but the worst-case quality of the resulting approximations is often unknown.

The SoS hierarchy was introduced in \cite{Shor87,Nesterov00,Parrilo00,Lasserre01} and reviewed in
\cite{Laurent09,Barak14}.  It is a family of SDP relaxations, parametrized
by the problem size $n$ and the level of the hierarchy $k$.  They run in time $n^{O(k)}$
and generally converge to the correct answer as $k\ra \infty$, but a crucial question is
to determine this rate of convergence. Over several domains, we know that
when $k$ is $\Omega(n)$, convergence is achieved: at $k=n$ exactly when the
domain is the boolean hypercube, and approximately when the domain is the
hypersphere~\cite{DW12,BGL17}. In this regime, the runtime of the algorithm is
comparable to that of brute-force search. However, in some cases, stronger positive results are known
for $k=O(1)$ or $k=O(\log n)$.

On the lower bound side, integrality gaps are known for which the SoS hierarchy, or
in some cases, more general families of SDPs, fail to give the correct answer.
Integrality gaps are known even for problems that are easy to solve, such as linear
equations over a finite field~\cite{Grigoriev01}.   Our first family of results can be
described as a set of integrality gaps for the DPS hierarchy.

\subsubsection*{Unique Games, Small-set Expansion and the 2-to-4 norm}
One further application of our results is to problems that are not obviously related to
quantum mechanics.  A central question in the theory of approximation algorithms is the
unique games conjecture (UGC, introduced in \cite{Khot02a} and reviewed in
\cite{Trevisan12}) which asserts the NP-hardness of a problem known as the ``unique games
problem.''  If true, the UGC would imply that level $k=1$ of the SoS hierarchy achieves
optimal approximation ratios for a wide range of problems.  There is a subexponential-time
algorithm for the unique games problem, and no $k=\omega(1)$ SoS integrality gap instances
are known.  While these cast doubt on the UGC, there is other evidence in favor, including
a conjectured $n^{\Omega(1)}$-round integrality gap and the fact that slight modifications
of the approximation parameters are known to yield NP-hard problems.

A close variant of the UGC is the small-set expansion (SSE) hypothesis which asserts that $\forall \eta>0, \exists \delta>0$ such that it is NP-hard to determine whether the maximum expansion of all subsets of fractional size $\delta$ is $\leq \eta$ or $\geq \eta$.  The SSE hypothesis implies UGC and indeed SSE is equivalent to a slightly restricted form of the unique games problem~\cite{RaghavendraST12}.   Here too no $k=\omega(1)$ integrality gaps are known.

The SSE problem in turn can be relaxed to the problem of estimating the $2\ra 4$ norm of a matrix.  If $A\in \R^{m\times n}$ then define the $2\ra 4$ norm of $A$ to be
$$ \|A\|_{2\rightarrow 4} := \max_{x\neq 0}\frac{\|Ax\|_4}{\|x\|_2}
\qquad\text{where}\qquad
\|y\|_p := \left(\sum_i |y_i|^p\right)^{1/p}.$$
(Here the number ``4'' can be replaced by any constant $q>2$.)

The $2\ra 4$ norm of a matrix is closely related to $h_{\Sep}$.  Indeed if we set
$M_{ij,kl} := \sum_a A_{ia}A_{ja}A_{ka}A_{la}$ for a real matrix $A$ then a
straightforward calculation shows that $h_{\Sep(d,d)}(M) = \|A\|_{2\ra 4}^4$, implying
that computing $\|A\|_{2\ra 4}$ reduces to computing $h_{\Sep}(\cdot)$.  Conversely
$h_{\Sep}$ can be reduced to calculating a $2\ra 4$ norm but this requires somewhat more
work~\cite{BHKSZ12}.

It was also shown in \cite{BHKSZ12} that SSE for a graph $G$ is approximately related to
the $2\ra 4$ norm of the projector onto the top eigenspace of $G$.  Thus, algorithms for
the $2\ra 4$ norm yield algorithms for SSE, and indeed known positive results for the SoS
hierarchy translate into subexponential time algorithms for SSE~\cite{BHKSZ12} (matching
those found using other methods).  On the other hand, the quasipolynomial hardness known
for the $2\ra 4$ norm does not necessarily imply similar hardness for SSE.  (This latter
hardness result assumes the Exponential-Time Hypothesis (ETH)~\cite{Impagliazzo:2001}
which asserts that 3-SAT instance with $n$ variables require $2^{\Omega(n)}$ time to
solve.)

Before our work, no $k=\omega(1)$ SoS integrality gap was known for the $2\ra 4$
norm, and finding such an integrality gap was proposed as an open problem by
Barak~\cite{Barak14}.  One of our main results is to establish SoS integrality
gaps for the $2\ra 4$ norm problem which roughly match the known computational
hardness results, but without needing the assumption of the ETH.  We also
establish weaker integrality gaps for a more general class of SDP relaxations,
called \SDPEF{}s.


\subsubsection*{Our contributions}

Our main contribution is to provide instances on which the above hierarchies (and others) fail to give the correct answer.  Thus we give lower bounds that do not rely on complexity-theoretic assumptions (and are in that sense ``unconditional''), although they do only apply to a specific family of algorithms. 
Additionally, our SoS bounds provide an explicit integrality gaps for the DPS hierarchy, while the ETH-based hardness result can only imply its existence.

Our results show (unconditionally) that the DPS hierarchy cannot estimate $h_{\Sep^2(d)}$ to constant accuracy at level $k$ unless $k\geq \tilde\Omega(\log d)$, corresponding to run-time $d^{\tilde\Omega(\log
  d)}$.  Similarly the $2\ra 4$ norm of $d$-dimensional matrices cannot be approximated to within a constant multiplicative factor  using $O(\log(d)/\poly\log\log(d))$ levels of SoS.  These results are described in Corollaries~\ref{cor:DPS} and \ref{cor:24}.
 This yields the first unconditional quantitative limits to the rate at
which the DPS hierarchy converges, and likewise the first $k=\omega(1)$ SoS integrality gaps for the $2\ra 4$ norm problem.

We also demonstrate integrality gaps corresponding to the other rows of \tabref{ETH}; 
In some cases these are nearly tight.  In 
\thmref{sos-lb-gnn}, we show that accuracy $\tilde O(1/d)$~\cite{GNN12,NOP09} requires $\tilde\Omega(d)$ levels of the SoS hierarchy, corresponding to exponential time.  Likewise for $h_{\Sep^{\sqrt{n}}(n)}$, constant-error approximations are shown in \thmref{sos-lb-multipartite} to require $\Omega(n)$ levels, which matches known SoS achievability results~\cite{BH-local}.


We further apply techniques due to Lee, Raghavendra, and Steurer~\cite{LRS15} to
extend our no-go theorems to a more general class of relaxations called \SDPEF{}s.  This is the first extension of LRS
that we are aware of to non-boolean domains.  To do this we need to work with
functions that are not strictly self-reducible, which requires some
modifications of the techniques of LRS.  However, while these bounds cover a
larger class of SDPs than the above bounds, they are quantitatively weaker.

\begin{theorem}[Informal, refer to Theorem~\ref{thm:lrs-gnn} and \secref{lrs}]
  Any \SDPEF{} of $h_{\Sep^2(d)}$ achieving 
  accuracy $1/\poly(d)$ must have a total number of variables
  $\geq d^{\tilde\Omega(\log d)}$.
\end{theorem}

Here ``\SDPEF{}'' is a technical condition defined in the body of
our paper. Roughly speaking an \SDPEF{} should replace the
optimization over separable states with an optimization over a convex
superset, which is defined by a semidefinite program.  (In fact we rule out a slightly larger class of
approximations.)

One corollary of these results  is an unconditional proof
of a version of the ``no approximate disentangler'' conjecture of Watrous~\footnote{The original goal of the conjecture was to rule
out a particular strategy for putting $\QMA(2)$ inside $\QMA$. Here
  $\QMA(2)$ is the set of languages where membership can be verified using a
  proof that is a pair of unentangled quantum states. } , for
which previously only the zero-error case was known~\cite{AaronsonBDFS08}.  This
conjecture asserts that if $\cN$ is a quantum channel from $D$ dimensions to
$d\times d$ dimensions such that $\Sep(d,d)\approx \text{Image}(\cN)$ (with
``$\approx$'' defined precisely in the body of our paper, but roughly speaking
it corresponds to $1/\poly(d)$ error) we must have $D\geq d^{\tilde
  \omega(\log(d))}$~\footnote{This is not the strongest possible version of the
conjecture since one could conceivably demand that $D\geq \exp(d)$; this
stronger form is false for the 1-LOCC norm~\cite{BCY} but is still an open
question for trace distance. }. 
  Our results imply
unconditional lower bounds on the input dimension $D$ albeit with 
parameters somewhat weaker than those based on ETH.

\begin{theorem}[Informal, refer to \thmref{disentangler}] 
  Let
  $D=\mathrm{dim}(\cH), d = \mathrm{dim}(\cK)$, and suppose that $\Lambda: D(\cH) \to D(\cK \otimes \cK)$ is an approximate disentangler with $1/\poly(d)$ error. Then
  \[ D \geq
d^{\tilde\Omega(\log(d)}). \]
\end{theorem}

\subsection{Entangled Games}

\subsubsection*{Noncommutative Polynomial Optimization and Entangled Games}

Another major class of optimization problems concerns polynomials in non-commuting
variables~\cite{HP-survey,PNA10}.  As we explain in \secref{prelim_sos}, these involve
optimizing operator-valued variables over a vector space of unbounded or even infinite
dimension.  One application is to understanding the set of ``quantum correlations'',
meaning the conditional probability distributions $p(x,y|a,b)$ achievable by local
measurements on a shared entangled states~\cite{NPA08}.  The most famous example of a
quantum-but-not-classical correlation was discovered by Bell in 1964~\cite{Bell64} and
gave a concrete experiment for which quantum mechanics predicts outcomes that are
incompatible with any theory that lacks entanglement or faster-than-light signaling.  More
recently, quantum correlations are studied in the context of multi-prover games where the
provers share entanglement (i.e., \emph{entangled games}); here, without a bound on the
dimension of the shared entangled state, one cannot rule out even infinite-dimensional
systems.  Non-commuting optimization can be useful even in cases where the dimension is
finite but exponentially large and the goal is to obtain smaller optimization problems,
e.g.~in quantum chemistry~\cite{Mazz04}.

A similar hierarchy was developed to approximate the set of
correlations achievable with local measurements of quantum
states~\cite{NPA08,DLTW08, BFS15}, or for more general polynomial optimization~\cite{PNA10}.  This is known variously as the
noncommutative Sum of Squares (ncSoS) hierarchy or the NPA
(Navascues-Pironio-Acin) hierarchy.  Here much less is known on either
the positive or negative side.  The ncSoS hierarchy similarly has
complexity increasing exponentially with $k$ and similarly converges
as $k\ra \infty$, although it is a famous open question (Tsirelson's
problem \cite{Tsirel-SW}) whether it indeed converges to the value of the best quantum strategy.

Computational hardness results are also known for the entangled value of quantum games.  If $\omega_{\text{entangled}}$ refers to the entangled value of a game with $k$ provers, one round, questions in $[Q]$, answers in a $O(1)$-sized alphabet, completeness $c$ and soundness $s$, then the known reductions from 3-SAT instances of size $n$ are described in Table~\ref{tab:omega_ent}.
However, no unconditional results are known for $k>5$.

\begin{table}[h]
\begin{center}
\begin{tabular}{ccccc}
\toprule
reference & $k$ &  $c$ & $s$ & $n$ \\
\midrule
\cite{KKMTV11} & 3 & 1 & $1-\frac{1}{\poly(Q)}$ & $O(Q)$ \\
\cite{IKM09} & 2 & 1 & $1-\frac{1}{\poly(Q)}$ & $O(Q)$ \\
\cite{IV12} & 4 & 1 & $2^{-Q^{\Omega(1)}}$ & $Q^{\Omega(1)}$ \\
\cite{Vidick13} & 3 & 1 & $2^{-Q^{\Omega(1)}}$ & $Q^{\Omega(1)}$ \\
\bottomrule
\end{tabular}
\caption{Hardness results for $\omega_{\text{entangled}}$ with $k$ provers and question alphabet size $Q$. 
\label{tab:omega_ent}}
\end{center}
\end{table}

Despite the lack of general results, specific solutions to the ncSoS hierarchy can be extremely useful.   For example, Tsirelson's 1980 outer bound on the winning probability of a particular quantum game~\cite{Tsirel, Chefles} has since had widespread application to topics including communication complexity~\cite{BBLMTU06}, ``self-testing'' quantum systems and multiparty
secure quantum computing~\cite{BP15,RUV13-short,RUV13-long}, and device-independent cryptography~\cite{BHK05}.  In quantum chemistry, the Pauli exclusion principle can be expressed as an operator inequality~\cite{CM15}, and far-reaching generalizations exist~\cite{AK08}, each of which can be seen as a dual feasible point for a ncSoS hierarchy.  While \cite{AK08} used representation theory to show the existence of these points in general, even explicit specific solutions can be useful~\cite{Ruskai07,Klyachko13}.

\subsubsection*{Our contributions}

We show the first known limitations on the ncSoS hierarchy for $k=\omega(1)$.  

\begin{theorem}[Informal, refer to Theorem \ref{thm:ncsos-nogo}]
There exists a sequence of two-player non-local games $G_n$ such that the entangled game value $\omega_{entangled}(G_n)\leq 1- c/n^2$ for some constant $c$ but the ncSOS hierarchy believes $\omega_{entangled}(G_n)=1$ up to level $m=\Omega(n)$. 
\end{theorem}

Previously the only unconditional lower bound on the error of the
ncSoS hierarchy for large constant $k$ was Slofstra's uncomputability
result~\cite{Slofstra16}, which implied (among other things) that the $\eps=0$ case could
not be solved by the $k^{\text{th}}$ level of the ncSoS hierarchy
unless $k$ is an uncomputably large function of the inputs.

\noindent  We are also able to give dimension lower bounds for any \SDPEF{} achieving reasonable accuracy, although these do not match our SoS lower bounds.

\begin{theorem}[Informal, refer to Theorem \ref{thm:game-nogo}]
There exists a sequence of two-player non-local games $G_n$ such that any \SDPEF{} approximating the entangled game value $\omega_{entangled}(G_n)$ to precision $O(1/n^2)$ has dimension $\geq n^{\log n / \poly\log \log n}$. 
\end{theorem}

As in the case of $h_{\Sep}$, we obtain SDP lower bounds by using the technique of~\cite{LRS15}. It is interesting to note that the games setting is somewhat more amenable to LRS's techniques than $h_{\Sep}$. This is because a key step in the proof of LRS requires embedding hard instances of the target problem into instances of the same problem with many more variables (this enables the use of random restrictions, which is crucial to the LRS proof). In the case of games, it is easy to embed a game with a smaller question alphabet into a game with a larger one---the referee simply ignores the extra questions. In contrast, this is not so simple for $h_{\Sep}$, since most QMA(2) protocols for CSPs involve sampling from a state which is a uniform superposition over all the variables, and most samples will not lie within a given small subset of the variables.

Our results here do not fully match the known computational hardness results.  In particular, the ETH-based arguments work at constant accuracy and our results rule out only approximations whose error decreases as a power of the number of questions and answers. This is because the reductions that yield constant accuracy ETH-conditional hardness involve operations like low-degree polynomial testing over higher order finite fields, which are not easy to make SoS complete. At the same time, it is worth noting that \thmref{game-nogo} extends to other SDPs for quantum correlations and in particular also limits the stronger hierarchy described in \cite{BFS15}.

\subsection{Technical Contributions}

\subsubsection*{Obtaining SoS hardness Results}
At a high level, similar to many other SoS hardness results, all of our SoS hardness
results stem from a classic result of Grigoriev~\cite{Grigoriev01},
showing hardness for the problem $\txor$ in the SoS model. 
To obtain integrality gaps for the problems considered here, we need to reduce (and sometimes embed) the classical hard problems in the SoS model to quantum problems, because our derivation of the SoS hardness of the 2-to-4 norm is inspired by its connection to the separability problem in quantum information. 

There have been several previous examples using reductions to prove
the hardness in the SoS model (e.g.,~\cite{Tul, OWWZ}). 
In \secref{lb_framework}, we  
formulate a framework of ``low-degree reductions'' which can be 
used in many cases.  It helps us to revisit previous NP-hardness
results and prove that they extend to
yield integrality gaps. 

We hope that this framework could facilitate the proof of
hardness in both the SoS model and the more general SDP model.
Another feature of this framework is that it can easily go beyond the problems over boolean domains. As far as we know,  our later application
of such a framework derives the first SoS hardness results for problems over \emph{non-boolean} domains. 

Let us look more closely at how this framework works. To obtain integrality gaps in the SoS model, one needs to show that (a) the SoS solution believes the value is large up to high level (degree), and (b) the true value is actually small. 
To achieve (a), we introduce the notion of \emph{low-degree reductions},
in which one requires the reductions preserve a SoS solution for the
reduced problems with almost the same value and a small amount of loss
of the degree.  These goals of preserving SoS solutions were referred
to as ``Vector Completeness'' in~\cite{Tul} and ``SoS Completeness''
in~\cite{OWWZ}.  In those papers this property was shown using
direct analysis of the SDP solutions.  Our approach
is instead to show that along with the NP-hardness reductions there
exist low-degree polynomials mapping solutions of the original problem
to solutions of the new problem.  The more complicated map on SDP
solutions then follows using generic arguments about low-degree reductions. 
  
To achieve (b), we resort to the study of quantum interactive proof protocols, which are connected to the separability problem, hence 2-to-4 norms, and at the same time provide the NP-hardness reductions shown in Table~\ref{tab:ETH}. 
In particular, we can import results such as the QMA(2) protocol for
3-SAT problem by Aaronson et al.~\cite{AaronsonBDFS08} and Chen-Drucker~\cite{CD10}, and make use of the \emph{soundness} of these reductions to achieve (b).  We follow the similar principle to import the two-player
non-local game protocol from Ref.~\cite{IKM09} (also in Table~\ref{tab:omega_ent}) to achieve (b) for SoS hardness results about entangled games. 

There is an issue with directly adopting Aaronson et al.'s QMA(2) protocol for 3-SAT as
the reduction, because the first step of this reduction makes use of the PCP theorem,
which turns out to be high-degree.  The use of the PCP theorem is to amplify the gap
between promised instances with some additional features. The SoS pseudo-solution for 3XOR
has already satisfied part of these features. We replace this step by a direct and
low-degree construction inspired by the proof of the PCP theorem.

As a minor technical contribution, we observe and make great use of the correspondence between density operators in quantum information and pseudo-expectations in SoS (e.g., Def.~\ref{def:conn}). This correspondence helps simplify a lot of calculations, especially given that our reductions are protocols from the study of quantum interactive proofs. 
We hope this connection can be found useful beyond the scope of this paper.

%

\subsubsection*{Extending to \SDPEF{}s}
Extending the above SoS hardness to the more general class of \SDPEF{}s crucially relies on the recent
breakthrough result by by Lee, Raghavendra and Steurer~\cite{LRS15}. Despite being by far
the most successful route to obtain unconditional bounds on general family of
SDP relaxations, the LRS result
has a few serious limitations that make it difficult to extend general SoS hardness
result (see Section~\ref{sec:LRS} for technical details.)

\begin{enumerate}
\item The result of LRS only works on the boolean cube $\{0,1\}^n$ and applies to CSP
  problems. Recent works by Braun et al~\cite{BPZ15} demonstrate the possibility of
  extending LRS's result to non-CSP problems (such as VertexCover, Max-Multi-CUT) but
  still over boolean domains, through {affine reductions} (or some relaxation of affine
  reductions discussed in~\cite{BPR16}).
   
  It is, however, not a priori clear whether this approach can be extended to non-boolean
  domains. Even for finite domains, the proof of~\cite{LRS15} may not automatically
  generalize as it makes intricate use of the structure of $\{0,1\}^n$, in particular
  Fourier analysis of functions over this domain. It is rather less clear how to extend
  the LRS approach to the continuous domains considered here---the hypersphere, for the
  separability problem, and infinite-dimensional Hilbert spaces, for the entangled game
  problem---without making use of, e.g., the theory of spherical harmonics.

  On top of that, most reductions used in deriving our SoS hardness are {low-degree}
  rather than {affine} (as in~\cite{BPZ15, BPR16}). Finally, the particular class of
  relaxations that~\cite{LRS15} applies to, here referred to as \SDPEF{}s, must satisfy
  certain rather rigid technical constraints, in particular one which we formulate as the
  {embedding} property.  

  Our crucial observation is that our reductions inspired by quantum protocols actually
  allow one to {embed} hard problems over $\{0,1\}^n$ into larger, even
  infinite-dimensional, domains. Moreover, our reductions are naturally low-degree, which
  could be a technically interesting point comparing to the affine reductions
  in~\cite{BPZ15, BPR16}.
As a result, we can avoid extending~\cite{LRS15}'s analysis to each setting, while still
extending its results on $\{0,1\}^n$ to much larger domains, albeit with some loss in parameters.
   \item The SDP lower bounds obtained by this result will never be better than
     quasi-polynomial in the problem size because of technical constraints. This
     limitation seems essential for the random restriction analysis that is
     central to the proof of LRS. (Recent work~\cite{KMR16} has succeeded in
     removing this limitation in the special case of LP relaxations, but doing
     so for the full class of SDP relaxations considered by LRS remains open.)
     A unfortunate consequence is that the lower bounds on general SDPs obtained via this method can be much looser than the SoS lower bounds they are based on. This partially explains why we have stronger SoS lower bounds than general SDP lower bounds in our results. 
   \item The LRS lower bound also crucially relies on a certain
     ``self-reducible'' structure of the problems. This is because of the use of pattern matrices in the proof of LRS. Intuitively, LRS can only show that some problem of size $n$ is hard if one can simultaneously embed a hard instance of size $m <n$ into each size-$m$ subset of the size-$n$ input. 
This is the case for CSP problems as well as for the problems considered in~\cite{BPZ15, BPR16}.
   
    However, because of our use of quantum protocols that involve superposed quantum states over all of the input variables, it is no longer true that a similar requirement will be satisfied for the quantum problems. As a result, we will only be able to obtain SDP lower bounds in \emph{some} of the cases where we have SoS lower bounds, and even in those cases, our parameters will be worse than those of the SoS results. 
\end{enumerate}

In summary, to the authors' knowledge, we are the first to extend LRS to problems over
non-boolean domains. However, due to the above limitations of LRS, there are gaps between
what we could get for SoS lower bounds and for SDP lower bounds. We consider it a major
open problem to improve on these techniques and prove tighter SDP lower bounds.

\subsubsection*{Handling SoS with non-commutative variables}

As our last technical contribution, we also demonstrate how to derive an ncSoS solution from a SoS solution for the purpose of showing the hardness of approximating entangled non-local game values. This step is necessary as all our SoS hardness results stem from the 3XOR problem, even for the optimization problem over non-commutative variables in the context of entangled games. As a result, without the following step, we only obtain a SoS hardness result rather than an ncSoS hardness result for entangled games. 

Our idea is to embed a SoS solution into an ncSoS solution, where the embedding is due to the connection between each question and a prover's corresponding strategy (i.e., operators). 
Intuitively the ncSoS solution constructed in this way can be said to ``cheat'' in the same way that classical SoS solutions do, and not by exploiting entanglement the way a ``valid'' quantum strategy would. We refer curious readers to the proof of Theorem~\ref{thm:ncsos-nogo} for details.

We summarize the complete reductions that handle all technical details
in Figure~\ref{fig:reductions}.

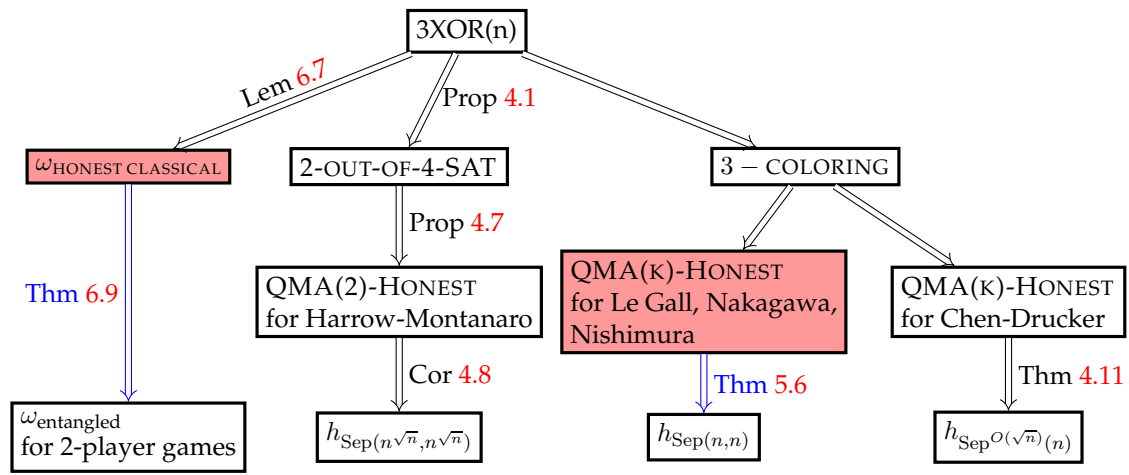
\begin{figure}[h]
\centering
\begin{tikzpicture}[scale = 0.9, transform shape]
  \tikzstyle{every node} = [rectangle, draw = black, very thick, align
  = left]
  \node (3XOR) at (5,6) {3XOR(n)};
  \node (tofsat) at (4,4) {\tofsat};
  \node (coloring) at (10, 4) {$3-\textsc{coloring}$};
  \node [fill =red!40] (game) at (0,4) {$\omega_{\textsc{honest classical}}$};
  \node (qma2cd) at (13, 2) {$\qmak$ \\ for Chen-Drucker};
  \node [fill = red!40] (qma2lnn) at (8.5, 2) {$\qmak$ \\ for Le Gall, Nakagawa, \\ Nishimura};
  \node (qma2hm) at (4, 2) {$\qmat$ \\ for Harrow-Montanaro};
  \node (hsep2) at (4, 0)  {$h_{\Sep(n^{\sqrt{n}}, n^{\sqrt{n}})}$};
  \node(hsep2onecopy) at (8.5,0) {$h_{\Sep(n, n)}$};
  \node(hsepmulti) at (13, 0) {$h_{\Sep^{O(\sqrt{n})}(n)}$};
  \node(gameentangled) at (0,0) {$\omega_{\text{entangled}}$ \\ for 2-player games};
  \draw  [-implies, double equal sign distance] (3XOR) -- (tofsat) node[midway, right, draw=none] {Prop~\ref{prop:reduction_txor_tofsat}};
  \draw  [-implies, double equal sign distance] (3XOR) -- (coloring); 
  \draw [-implies, double equal sign distance] (3XOR) -- (game) node[midway, sloped, above, draw=none] {Lem~\ref{lem:sos_lb_games}};
  \draw [-implies, double equal sign distance] (coloring) -- (qma2cd);
  \draw [-implies, double equal sign distance] (coloring) -- (qma2lnn);
  \draw [-implies, double equal sign distance] (tofsat) -- (qma2hm) node[midway, right, draw=none] {Prop~\ref{prop:witness_qmat}};
  \draw [-implies, double equal sign distance] (qma2hm) -- (hsep2) node[midway, right, draw=none] {Cor~\ref{cor:DPS}};
  \draw [-implies, double equal sign distance] (qma2cd) -- (hsepmulti) node[midway, right, draw=none] {Thm~\ref{thm:sos-lb-multipartite}};
  \draw [-implies, double equal sign distance, blue] (qma2lnn) -- (hsep2onecopy) node[midway, right, draw=none] {Thm~\ref{thm:lrs-gnn}};
  \draw [-implies, double equal sign distance, blue] (game) --  (gameentangled) node[midway, left, draw=none] {Thm~\ref{thm:game-nogo}};
\end{tikzpicture}
\caption{All our results are derived by applying a series of
  low-degree reductions to the integrality gap for 3XOR given
  by~\cite{Grigoriev01}. The red nodes indicate problems over the boolean cube to which the LRS theorem is applied. The blue arrows are ``embedding reductions'' that map problems over the boolean cube to problems over general domains, e.g. the set of separable states or the set of quantum entangled strategies.}
\label{fig:reductions}
\end{figure}

\subsection{Open Problems}

 Our work leaves open a number of
intriguing questions.
\begin{itemize}
\item Our results can be viewed as an extension of \cite{LRS15} beyond $\{0,1\}^n$ to sets
  such as the hypersphere.  However, as explained in the above, we face a lot of
  difficulties in doing so because of the limitations of LRS.  We consider it a major open
  problem to improve on these techniques and prove tighter SDP lower bounds over general
  domains. In particular, we believe that the quantum problems considered in
  this paper are a good motivation to prove extension complexity lower bounds
  for CSPs that are exponential, and not just superpolynomial. Progress in this
  direction has been achieved for LP extension complexity~\cite{KMR16}, but the
  SDP case remains open.
\item It raises the motivation to examine convex but non-SDP-based relaxations, such as
  the entropic bounds used in \cite{PH11}.
\item Our integrality gaps for the noncommutative hierarchies involve games whose
  entangled values cannot be effectively bounded using low-degree SoS proofs, but can be
  bounded using other methods.  These methods usually rely on specific features of the
  game; e.g.~consistency checks or other tests~\cite{IKM09,IV12}.  Can these upper-bound
  methods be understood more generally?  For example, can they be described using a high
  level of the SoS hierarchy?  A motivating example is 3XOR, where Gaussian elimination
  can be used to refute unsatisfiable instances once we get to level $n$ of the
  hierarchy~\cite{Grigoriev01}.  We note that, contrary to Barak's ``Marley
  principle''~\cite{Barak14}, the soundness bounds in \cite{IKM09,IV12} neither rely on
  the probabilistic method nor low-degree polynomials.
\item Our results on
quantum correlations should be strengthened to match the ETH-based
bounds (e.g.~\cite{IV12}). We also give new motivation for proving that random 3XOR
has low entangled value (cf. section 2.3 of \cite{Pala15}).  If known, this would give a much more direct and efficient no-go result for games.
\item We have constructed entangled states that appear
  separable to the lower levels of the DPS hierarchy. But are these
  states generic in the sense of \cite{Lancien15}?  We believe that
  this is the case; i.e.~a random state within the convex set accepted
  by DPS is likely to have distance from separable states that is
  comparable or better to our examples.  However, different techniques
  will be needed to prove this.
\end{itemize}

\section{Preliminaries}  \label{sec:prelim}
We provide a brief introduction to sum-of-squares proofs/optimization
in Section~\ref{sec:prelim_sos} and then summarize relevant background about quantum information and our terminology in Section~\ref{sec:prelim_qi}.
 
\subsection{Polynomial Optimization and Sum-of-Squares Proofs} \label{sec:prelim_sos}
In this section, we lay out the basics of the sum-of-squares (SoS) optimization algorithms.  They were introduced in \cite{Shor87,Nesterov00,Parrilo00,Lasserre01} and reviewed in \cite{Laurent09,Barak14}.

\paragraph{Symmetric subspace}
For any vector space $V$ and positive integer $k$, the permutation
group on $k$ letters $S_k$ acts on the vector space $V^{\ot k}$ by
$ P_{\pi} \ket{i_1} \ot  \dots \ot \ket{i_k} = \ket{i_{\pi^{-1}(1)}} \ot
\dots \ot\ket{ i_{\pi^{-1}(k)}}$. The $k$-partite symmetric subspace $\Sym^{k}V$ is
the subspace of $V^{\ot k}$ that is invariant under this action:
\[ \Sym^{k}V = \{ \ket v \in V^{\ot k} : P_{\pi}\ket v =\ket v \; \forall \pi \in
  S_k\}. \]

\paragraph{Polynomials and non-commutative polynomials.}  Let $\R[x]
:= \R[x_1,\ldots,x_n]$ be the set of real-valued polynomials over $n$
variables, and let $\R[x]_d$ be the subspace of polynomials of degree
$\leq d$.  We also define $\R\ipp{X}$ to be the set of {\em
  non-commutative polynomials} in $X_1,\ldots,X_n$, which we think of
as Hermitian operators that do not necessarily commute, and which we
call ``nc polynomials'' for short.   The
nc polynomials of degree $\leq d$ are denoted
$\R\ipp{X}_d$ and are isomorphic to $\bigoplus_{d'\leq d}
(\R^n)^{\otimes d'}$, while the ordinary commutative polynomials
$\R[x]_d$ can be viewed as $\bigoplus_{d'\leq d'} \Sym^{d'}\R^n$,
where $\Sym^{d'}V$ denotes the symmetric subspace of $V^{\otimes
  d'}$.
These notational conventions follow the optimization literature with
the roles of $n$ and $d$ reversed from the way they are usually used
in quantum information theory.

\paragraph{Polynomial optimization.}  Given polynomials $f,g_1,\ldots,g_m \in \R[x]$, the basic polynomial optimization problem is to find 
\be f_{\max} := \sup_{x\in\R^n} f(x) \text{ subject to } g_1(x)= \cdots = g_m(x)=0.
\label{eq:poly-opt}\ee
Equivalently we could impose inequality constraints of the form $g_i'(x)\geq 0$ but we will not explore this option here. In the non-commutative setting we need to optimize over both variables and a density.  Given $F,G_1,\ldots,G_m \in \R\ipp{X}$, define
\be F_{\max} := \sup_{\rho,X = (X_1,\ldots,X_n)} \tr[\rho F(X)] \text{ subject to }\rho\geq 0, \tr\rho=1, G_1(X)=\cdots=G_m(X)=0.
\label{eq:nc-poly-opt}\ee
Note that the supremum here is over density operators $\rho$ and Hermitian operators $X_1,\ldots,X_n$ that may be infinite dimensional; see \cite{Tsirel-SW} for a discussion of some of the mathematical difficulties here.

\paragraph{Sum-of-Squares (SoS) proofs.} Although \eq{poly-opt} and \eq{nc-poly-opt}
are in general $\NP$-hard to compute exactly, the SoS hierarchy is a
general method for approximating $f_{\max}$ or $F_{\max}$ from above.
This complements simply guessing values of $x$ or $(\rho,X)$ which
provides lower bounds on $f_{\max}$ or $F_{\max}$ when they satisfy
the constraints.  Focusing for now on the commuting case, a SoS proof is a bound that makes use of the fact that $p(x)^2\geq 0$ for any $p\in \R[x]$.   In particular, a SoS proof that $f(x)\leq c$ for all valid $f$ is a collection of polynomials $p_1,\ldots,p_k, q_1,\ldots,q_m\in \R[x]$ such that
\be c-f = \sum_{i=1}^k p_i^2  +\sum_{i=1}^m q_ig_i. \label{eq:SoS-proof}\ee
Observe that the RHS is $\geq 0$ when evaluated on any $x$ satisfying
$g_i(x)=0, \,\forall i$; for this reason, we refer to \eq{SoS-proof}
as a Sum-of-Squares (SoS) proof.  In particular, it is a proof that $c-f(x)\geq 0$ whenever $g_i(x)=0$ for all $i$.
This is a degree-$d$ SoS proof if each term $p_i^2$ and $q_ig_i$ is in $\R[x]_d$.  Finding an SoS proof of degree $\leq d$ can be done in time $n^{O(d)}m^{O(1)}$ using semidefinite programming~\cite{Laurent09}. 

If we find the minimum $c$ for which \eq{SoS-proof} holds, then we obtain a hierarchy of upper bounds on $f_{\max}$, referred to as the SoS hierarchy or the Lasserre hierarchy.  Denote this upper bound by $f_{\text{SoS}}^d$.   
Given mild assumptions on the constraints $g_1,\ldots,g_m$ one can prove that $\lim_{d\ra\infty}f_{\text{SoS}}^d = f_{\max}$~\cite{Laurent09}.  The tradeoff between degree $d$ and error ($f_{\text{SoS}}^d-f_{\max}$) is the key question about the SoS hierarchy.  We can also express this tradeoff by defining $\degsos(c-f)$ to be the minimum $d$ for which we can find a solution to \eq{SoS-proof}.  Note that $\degsos$ has an implicit dependence on the $g_1,\ldots,g_m$.

A non-commutative SoS proof can be expressed similarly as
\be c-F = \sum_{i=1}^k P_i^\dag P_i + \sum_{i=1}^m Q_i G_i R_i,\ee
for $\{P_i\}, \{Q_i\}, \{R_i\} \subset \R\ipp{X}$.  Likewise the best degree-$d$ ncSoS (noncommutative SoS) proof can be found in time $n^{O(d)}m^{O(1)}$, and we denote the corresponding value by $F_{\text{SoS}}^d$.  It is known that  $F_{\max} \leq F_{\text{SoS}}^d$ for all $d$ and $\lim_{d\ra\infty}F_{\text{SoS}}^d = F_{\max}$~\cite{HM03}.

\paragraph{Pseudo-expectations.}  We will work primarily with a dual version of SoS proofs that have an appealing probabilistic interpretation.  A degree-$d$ pseudo-expectation $\psE$ is an element of $\R[x]_d^*$ (i.e.~a linear map from $\R[x]_d$ to $\R$) satisfying
\bit
\item {\bf Normalization}. $\psE[1]=1$.
\item {\bf Positivity}. $\psE[p^2]\geq 0$ for any $p\in \R[x]_{d/2}$.
\eit
We further say that $\psE$ satisfies the constraints $g_1,\ldots,g_m$ if $\psE[g_iq]=0$ for all $i\in [n]$ and all $q\in \R[x]_{d-\deg(g_i)}$.
Then SDP duality implies that 
\be f_{\text{SoS}}^d = \max \{\psE[f] : \psE\text{ is a degree-$d$ pseudo-expectation satisfying } g_1,\ldots,g_m\}.\ee

The term ``pseudo-expectation'' comes from the fact that for any
distribution $\mu$ over $R^n$ we can define a pseudo-expectation
$\psE[f] := \E_{x\sim\mu}[f(x)]$.  Thus the set of pseudo-expectations
can be thought of as the low-order moments that could come from a
``true'' distribution $\mu$ or could come from a ``fake''
distribution. Indeed an alternate approach (which we will use only in \lemref{poly_map_pseudo_distribution}) proceeds from defining ``pseudo-distributions'' that violate the nonnegativity condition of probability distributions but in a way that cannot be detected by looking at the expectation of polynomials of degree $\leq d$~\cite{LRS15}.
We can define a noncommutative pseudo-expectation $\psE\in \R\ipp{X}_d^*$ similarly by the constraints $\psE[1]=1$ and $\psE[p^\dag p]\geq 0$ for all $p\in \R\ipp{X}_{d/2}$.

\paragraph{The boolean cube.}
Throughout this work, we will be interested in the special case of pseudo-expectations over the boolean cube $\zo^n$.  This set is defined by the constraints $x_i^2-1=0$, $i=1,\ldots,n$, and thus we say that $\psE$ is a degree-$d$ pseudo-expectation over $\zo^n$ if for any variable $x_i$ and polynomial $q$ of degree at most $d - 2$,
 \be \psE[(x_i^2-1) q] = 0.\label{eq:boolean-psE}\ee
This means we can define $\psE$ entirely in terms of its action on multilinear polynomials. 
 
\subsection{Quantum Information}  \label{sec:prelim_qi}
\begin{trivlist}

\item \textbf{Quantum States.}  The state space $\cA$ of  $m$-qubit states is the complex Euclidean space $\complex^{2^m}$.
An $m$-qubit quantum state is represented by a density operator $\rho$, i.e., a positive semidefinite matrix with trace 1, over $\cA$. The set of all quantum states
in $\cA$ is denoted by $\density{\cA}$. 
A quantum state $\rho$ is called a \emph{pure} state if $\rank(\rho)=1$; otherwise, $\rho$ is called a \emph{mixed} state.  An $m$-qubit pure state $\rho=\ket{\psi}\bra{\psi}$, where $\ket{\psi}$ is a unit vector in $\complex^{2^m}$. (We might abuse $\ket{\psi}$ for $\ket{\psi}\bra{\psi}$ when it is clear from the context.)
The Hilbert-Schmidt inner product on the operator space $\lin{\cA}$ is defined by $\ip{X}{Y}=\tr (X^\dag Y)$  for all $X,Y \in \lin{\cA}$, where $\dag$ is the adjoint operator.

An important operation that can be performed on bipartite or
multipartite states is the \emph{partial trace}. For a bipartite state
$\rho \in \density{\cH_A \otimes \cH_B}$, the partial trace over system
$A$ is a density matrix $\rho_B := \Tr_{A}[\rho]$, whose matrix
elements are given by
\[ \bra{i} \rho_{B} \ket{j} = \sum_{k} \bra{k}\otimes\bra{i} \rho
\ket{k} \otimes \ket{j}. \]
One can analogously compute the partial trace of a multipartite state
over any subset of its component subsystems. The state obtained by
partial tracing some of the subsystems is called the \emph{reduced}
state on the remaining subsystems.

Let $\Sigma$ be a finite nonempty set of \emph{measurement outcomes}. A {\em positive-operator valued measure (POVM)} on the state space $\cA$
with outcomes in $\Sigma$ is a collection of positive semidefinite
operators $\set{P_a : a \in \Sigma}$ such that $\sum_{a\in \Sigma} P_a=\I_\cA$.
When $P_a^2=P_a$ for all $a \in \Sigma$, such a POVM is called \emph{projective} measurement. 
When this POVM is applied to a quantum state $\rho$, the probability of each outcome $a \in \Sigma$ is  $\ip{\rho}{P_a}$. When outcome $a$ is observed, the quantum state $\rho$ becomes the state $\sqrt{P_a}\rho\sqrt{P_a} / \ip{\rho}{P_a}$.

\item \textbf{Distance Measures.} For any $X \in \lin{\cA}$ with singular values $\sigma_1,\cdots, \sigma_d$, where $d=\dim(\cA)$, the trace norm of $\cA$ is $\trnorm{X}=\sum_{i=1}^d \sigma_i$.
The \emph{trace distance} between two quantum states $\rho_0$ and
$\rho_1$ is defined to be 
\[
\frac{1}{2}\trnorm{\rho_0-\rho_1}.\] 

\item \textbf{Multipartite states and separability}
Suppose $\cH_{A}$ and $\cH_{B}$ are two state spaces with dimensions
$d_A, d_B$, and consider the bipartite state space $\cH_{A} \otimes
\cH_{B}$ obtained by taking their tensor product. It is clear that for any $\rho_A \in \density{\cH_{A}}$ and
$\rho_B \in \density{\cH_{B}}$, $\rho_A \otimes \rho_B$ is a density
matrix over $\cH_{A} \otimes \cH_{B}$. However, most density matrices
in $\density{\cH_{A} \otimes \cH_{B}}$ \emph{cannot} be written in
this form as a tensor product. We call any state in the convex hull of the tensor
product states \emph{separable}, and all other states
\emph{entangled}. More formally, define the set of separable states to be
\[ \Sep(d_A, d_B) = \conv(\{\rho \otimes \sigma: \rho \in \cH_{A},
\sigma \in \cH_{B}\}). \]
The problem $h_{\Sep(d_A, d_B)}(M)$ is, given Hermitian $M$ acting
over $\cH_{A} \otimes \cH_{B}$ with $\|M\| \leq 1$, compute
\[ h_{\Sep(d_A, d_B)}(M) := \max_{\rho \in \Sep(d_A, d_B)} \Tr[M
\rho]. \]
It is easy to see that the optimum value will be achieved on the extreme points of the separable states, which 
are simply the pure product states (i.e., $\ket{\psi_A}\bra{\psi_A} \ot \ket{\psi_B}\bra{\psi_B}$). 
We can interpret this as searching for the separable state that has
the highest chance of being accepted by the POVM measurement $\{M, \I
- M\}$. For simplicity, in the rest of the paper we will specialize to
the case where $d_A = d_B$; it is easy to reduce general instances of
the problem to this case~\cite{HM13}.

More generally, we can consider state spaces that are built of tensor
products of many factors, i.e. $\cH = (\cH_A)^{\otimes k}$, and define
the set of $k$-partite separable states as
\[ \Sep^k(d_A) = \conv(\{\rho_1 \otimes \rho_2 \otimes \dots \otimes
\rho_k : \rho_1, \dots, \rho_k \in \density{\cH_A}\}).\]
One can likewise consider the multipartite separability problem
$h_{\Sep^k(d_A)}(M)$.
\end{trivlist}

\subsection{SoS hierarchies for Quantum Problems}

As noted in the introduction, the problem $h_{\Sep}$ is an instance of
polynomial optimization: in the bipartite case, the objective function
is a degree-$4$ polynomial while the constraints are degree-$2$
polynomials in the coefficients of the state vector. It is possible to
derive an SoS hierarchy for this problem by starting with this
formulation and applying the procedure described in
Subsection~\ref{sec:prelim_sos}. However, for the separability problem
it is more convenient to consider an equivalent formulation couched
in the language of quantum information, called the DPS
hierarchy~\cite{DPS03}. The value of the $k$-th level of DPS is given
by $\mathrm{DPS}_k(M) = \max_{\rho \in k\mathrm{-ExtPPT}} \Tr[\rho M]$,
where the set $k\mathrm{-ExtPPT}$ is the set of ``$k$-extendable
PPT states'': all bipartite states $\rho$ that can be obtained as the
\emph{reduced} state of a $k$-partite state $\rho'$, whose support
lies entirely in the $k$-partite symmetric subspace, and whose \emph{partial
  transpose} (interchange of row and column indices) over any subset
of subsystems is positive semidefinite. This equivalence between this
formulation and the usual formulation of SoS can be seen by
interpreting the entries of $\rho$ as the pseudo-expectations of
degree-$k$ monomials of the coefficients of the state vector; see \cite{DPS03} or
\cite{BHKSZ12} for more discussion of this point.

In the context of entangled games, one can follow a similar principle and develop an SoS hierarchy for non-commutative variables to approximate the set of
correlations achievable with local measurements of quantum
states~\cite{NPA08,DLTW08, BFS15}.  
Here, the usual formulation of $k$-the level of ncSoS corresponds to the non-local correlations achievable by applying a degree-$k$ monomial operator (i.e., measurement) on the shared quantum state. 


\section{Framework of Deriving Lower Bounds} \label{sec:lb_framework}
In this section, we demonstrate our framework of deriving
sum-of-squares (SoS) or semidefinite programming (SDP) lower bounds
for optimization problems.  To this end, we formalize the familiar notions of
optimization problem, \SDPEF{}s and integrality gaps.  Then we
show general methods for reducing optimization problems to each other
as well as mapping integrality gaps for one problem/relaxation pair to
another.  

\subsection{Optimization problems and integrality gaps}

We formulate the following abstract definition of optimization
problem.  This definition does not address the computational
difficulty of {\em solving} the problem, which can often be NP-hard
or even uncomputable. 

\begin{definition}[Optimization Problem] \label{def:opt}
An optimization problem $A$,
 denoted by $\Delta^A =
\{\Delta_n^A\}_{n\in \N}$,  
is a family of
collections of optimization instances that are parameterized by the instance
size $n \in
\mathbb{N}$ and 
which consists of following components:
\begin{itemize}
  \item \textbf{Feasible Set}: $\cP^A_n$ is the set of feasible solutions.
  \item \textbf{Instances}: $\Delta^A_n$ is the set of instances (or
    objective functions), each of which is a
    map  $\Phi:\cP^A_n \mapsto [0,1]$.
  \item \textbf{Optimum Value:} Given $n$ and  $\Phi \in
    \Delta^A_n$, the optimum value of the instance $\Phi$ is
   \[
      \OPT(\Phi) :=\max_{x \in \cP^A_n} \Phi(x).
   \]
\end{itemize}
\end{definition}
In general we will choose our parametrization so that the number of variables in
the problem is equal to $n$ or to a polynomial function of $n$. Note that under this defintion, the feasible set is allowed to depend only on the
instance size $n$ and not on the instance itself. An example of an optimization
problem that fits under this definition is MAX-CUT, in which $n$ is the
number of vertices in a graph, $\cP^{\text{MAX-CUT}}_n = \{0,1\}^n$ and $\Delta^{\text{MAX-CUT}}_n$
is the set of functions of the form $\Phi(x) := \E_{(i,j)\sim E}
(x_i-x_j)^2$ for some $E \subset [n] \times [n]$.  As we can see from
the example, the functions $\Phi$ can usually be efficiently
specified (in this case by the edge set $E$), and can be thought of as
the computational ``question'' while the optimal value of $x$ can be
thought of as the ``answer.''

We will focus on the following important special cases of optimization problems.

\begin{definition}[Polynomial Optimization] \label{def:pol_opt}
A \emph{polynomial} optimization problem $A$ is an optimization
problem in which the feasible set $\cP_n^A$ is a variety of 
$\mathbb{R}^{n}$ defined by $m$ polynomial constraints for some
bounded function $m=m(n)$ and in which each instance $\Phi \in \Delta_n^A$ is a
polynomial function from
$\cP^A_n$ to $[0,1]$.   Here ``variety'' 
means that
$$\cP_n^A = \{x  \in \R^n : g_1(x)=\cdots = g_m(x)= 0\},$$
for some polynomials $g_1,\ldots,g_m$.
\end{definition}

\begin{definition}[Boolean Polynomial Optimization] \label{def:bol_opt}
A \emph{boolean polynomial} optimization problem $A$, denoted by
$\Pi^A$, is a polynomial optimization problem defined by the constraints $x_i^2=1$ for $i=1,\ldots,n$.  Thus the feasible set
$\cP^A_n$ for such a problem is always the boolean hypercube $\cP^{\text{bool}}_n
= \zo^n$.  
\end{definition}
It is easy to see that the above definitions of optimization problems
capture many problems of interest. For example, MAX-3-SAT, MAX-CUT and
other MAX-CSPs can be formulated as boolean polynomial optimization problems
with the objective function being a polynomial that counts the
fraction of satisfiable clauses, as indicated below.
\begin{definition}[Constraint Satisfaction Problem] \label{def:csp}
A (maximum) \emph{constraint satisfaction problem} ((MAX)-CSP) $A$ is a type of
optimization problem over the boolean hypercube $\zo^n$ specified by a
collection of \emph{clauses} $\{h_1, \dots, h_m\}$, where each clause
$h_i: \zo^n \to \{0,1\}$ is a boolean function, and the objective function $f =
\frac{1}{m} \sum_{i=1}^{m} h_i$ counts the fraction of clauses that evaluate to
$1$. 
\end{definition}
\begin{proposition}\label{prop:csp_to_poly}
Any CSP where each constraint depends on $\leq \kappa$ variables can be written as a boolean polynomial optimization problem, where the objective function is a polynomial of degree $\kappa$.
\end{proposition}

\begin{proof}
This is a simple example of polynomial interpolation.  Assume for ease
of notation that each constraint depends on exactly $\kappa$ variables.
  We first show that each constraint can be expressed as a low-degree polynomial. Given a string $(y_1, \dots, y_\kappa) \in \{\pm 1\}^{\kappa}$, we define the \emph{indicator} function
  \[ \mathbf{1}_{y_1, \dots , y_\kappa}(x_1, \dots, x_\kappa) =
  \prod_{i=1}^{\kappa} \frac{1+x_iy_i}{2} . \]
  This function is a polynomial of degree $\kappa$. It is easy to see that when $x_i = y_i$ for all $i$, then $\mathbf{1}_{y_1, \dots, y_\kappa}(x_1, \dots, x_\kappa) = 1$, and if  $x_i \neq y_i$ for any $i$, then $\mathbf{1}_{y_1, \dots, y_\kappa}(x_1, \dots, x_\kappa) = 0$. Using these indicator functions, we can express any boolean function $h$ over $\kappa$ variables 
  as a polynomial with degree $\kappa$:
  \[ h(x_1, \dots, x_\kappa) = \sum_{(y_1, \dots y_\kappa) \in \{\pm 1\}^\kappa}  h(y_1, \dots,y_\kappa) \mathbf{1}_{y_1, \dots, y_k} (x_1, \dots, x_\kappa). \]
  Thus, if we are given a CSP with clauses $\{h_1, \dots, h_m\}$, each of which
  depends on $\kappa$ variables, then the total objective function $f = \frac{1}{m}
  \sum_{i=1}^m h_i$ is a polynomial of degree $\kappa$. 
\end{proof}

Another important class of optimization
problems are operator norms of linear functions, defined as follows.
If $\cA,\cB$ are normed spaces and $T:\cA\ra \cB$ is a linear map then
$\|T\|_{\cA\ra \cB} := \sup_{x\neq 0} \|T(x)\|_{\cB} / \|x\|_{\cA}$ can be
thought of as an optimization problem where $\cP_n$ is the unit ball of
$\cA$ and $\Phi_n^A(x) = \|T(x)\|_{\cB}$.  If $\cA = \ell_p^n$ and $\cB
=\ell_q^m$ then this corresponds to a polynomial optimization problem.
Computing $h_{\Sep}$ (cf.~\eq{hSep}) can be similarly be
formulated as a polynomial optimization problem, where the feasible
set is the unit sphere~\cite{BHKSZ12}.

Our goal is to find optimization problems where the SoS hierarchy and other \SDPEF{}s fail.  These examples are known as ``integrality gaps,'' where the terminology comes from the idea of approximating integer programs with convex relaxations. For our purposes, an integrality gap will be an example of an optimization problem in which the true answer is lower than the output of the \SDPEF{}.  To achieve this, we need to demonstrate a feasible point of the SDP with a value that is larger than the true answer.  These feasible points are called {\em pseudo-solutions}, and we will define them for any polynomial optimization problem as follows.

\begin{definition}[Pseudo-Solution] \label{def:pseudo_sol}
Let $A$ be a polynomial optimization problem. Let $\Phi^A_n\in \Delta^A_n$ be an instance of optimization $A$ for some $n$. A {\em degree-$d$ value-$c$ pseudo-solution} for $\Phi^A_n$ is a degree-$d$ pseudo-expectation $\psE$ satisfying the constraints of $\cP^A_n$
 such that
\[
   \psE[\Phi^A_n(x)] \geq c
\]
\end{definition}

In the case of CSPs, we can also define a stronger type of
pseudo-solution that not only achieves an objective value of $1$, but
also satisfies the constraints of the CSP as ``hard'' constraints.  
\begin{definition}\label{def:hard_constraints}
Given a CSP $A$ whose objective function is a polynomial of degree $k$, we say a degree-$d$ pseudo-expectation $\psE[\cdot]$ \emph{perfectly satisfies} $A$ if for every constraint $g_i(x)$ of $A$, and every polynomial $p(x)$ with $\deg(p) \leq d - k$, 
\[ \psE[ p(x) (g_i(x) - 1)] = 0. \]
\end{definition}

A single degree-$d$ value-$c$ pseudo-solution for an instance
$\Phi^A_n$ implies that the sum-of-squares approach (up to degree $d$) believes the optimum value of $\Phi^A_n$ is at least $c$. If the true optimum value of $\Phi^A_n$ is smaller than $c$, then such a pseudo-solution serves as an integrality gap for the SoS approach, i.e.~an example where the SoS hierarchy gives the wrong answer.
To refute the power of the SoS hierarchy, we need to establish such pseudo-solutions as well as small true optimum values for any large $n$. 

\begin{definition}[Integrality gap] \label{def:SoS_lb_wit}
Let $A$ be any polynomial optimization problem. Let $d=d(n),c=c(n),s=s(n)$ be functions of $n$ such that $0\leq s< c \leq 1$. A degree-$d$ value-$(c,s)$ integrality gap for $A$ is a collection of $\Phi^A_n \in \Delta^A_n$ for each $n\geq n_0$, s.t. 
\begin{itemize}
 \item The true optimum value $\OPT(\Phi^A_n)\leq s$.
 \item For each $n \geq n_0$, there exists a degree-$d$ value-$c$ pseudo-expectation $\psE_n$ for $\Phi^A_n$ such that $ \psE_n[\Phi^A_n(x)] \geq c.$
\end{itemize} 
\end{definition}

\noindent We can relate integrality gaps to lower bounds on $\degsos$ as follows.

\begin{proposition} \label{prop:sos_lb}
Let $A$ be any polynomial optimization problem with a degree-$d$
value-$(c,s)$ integrality gap. For some $0<\delta \leq c-s$ let $f_n=
c -\delta - \Phi^A_n (x^A_n)$  where $\Phi^A_n$ is from the
integrality gap and $x^A_n \in  \cP^A_n$. Then $f_n$ is a polynomial
taking nonnegative values over $\cP^A_n$ and has  $\degsos(f_n)> d$. 
\end{proposition}

\begin{proof}
Immediate from the definitions.
\end{proof}

\subsection{Reduction between optimization problems}
To obtain SoS lower bounds for optimization problems, it suffices to
establish integrality gaps. However, it is not clear how to obtain
such integrality gaps in general, which might be a challenging task on
its own. 
Here, we formulate an approach to establish such integrality gaps through reductions. Specifically, we start with some optimization problem with known integrality gaps and reduce it to an optimization problem that we want to establish integrality gaps. 

\begin{definition}[Reductions] \label{def:reduction}
A \emph{reduction} $R_{A \Rightarrow B}$ from optimization problem $A$ to
optimization problem $B$ is a map from $\Delta^A$ to $\Delta^B$;
i.e. $R(\Phi_{n_A}^A)\in \Delta^B_{n_B}$, where $n_B$ is a function of $n_A$.
\end{definition}

We remark that for the purpose of establishing integrality gaps, the reduction needs to be neither explicit or efficient. However, it is favorable to have the following properties for the reduction. 

\begin{definition}[Properties of Reductions] \label{def:reduction_prop}
A \emph{reduction} $R_{A \Rightarrow B}$ from optimization problem $A$ to optimization problem $B$ is called 
\begin{itemize}
  \item \textbf{$(s^B,s^A)$-approximate}\footnote{We write the parameters in
      this order in order to match the convention
$(c,s)$-approximate extended formulations, which are defined below. We will see
that a $(s^B, s^A)$-approximate reduction gives rise to a $(c=s^B, s
=s^A)$-approximate extended formulation.} if for any $n$ and any $\Phi^A_{n_A}$ and its corresponding $\Phi^B_{n_B}=R(\Phi^A_{n_A})$, we have 
  \[
     \OPT(\Phi^A_{n_A}) = \max_{x \in \cP^A_{n_A}} \Phi^A_{n_A}(x) \leq s^A 
\Rightarrow  \OPT(\Phi^B_{n_B}) = \max_{x \in \cP^B_{n_B}} \Phi^B_{n_B}(x) \leq s^B.
  \]
Here $s^A,s^B$ are understood to be functions of $n_A, n_B$, and $s^A \leq s^B$.
  \item \textbf{embedded} if for any $n_A$, there is an additional map $E:\cP^A_{n_A} \mapsto \cP^B_{n_B}$ such that for any $\Phi^A_{n_A}$ and its corresponding $\Phi^B_{n_B}$, any $x^A_{n_A} \in \cP^A_{n_A}$ and its corresponding $x^B_{n_B} = E(x^A_{n_A}) \in \cP^B_{n_B}$, we have 
  \[
     \Phi^A_{n_A}(x^A_{n_A})=\Phi^B_{n_B}(x^B_{n_B}).
  \]
\end{itemize}
\end{definition}

The first property shows the soundness of the reduction, while the second
property can be viewed as a strong statement about completeness.  Not only
should the optimum value of the reduced problem be at least as large, but each
$x\in\cP^A_n$ (i.e.~including non-optimal $x$) corresponds to some point in
$\cP^B_n$ with the same value under $\Phi^B_n$.  This condition was needed for
the recent SDP lower bounds in \cite{LRS15}.

For reductions between polynomial optimization problems,  the following property is crucial to establish pseudo-solutions, and eventually integrality gaps, for the reduced problems. 

\begin{definition}[Pseudo-solution Preserving Reduction] \label{def:reduction_pseudo_sol}
Let $R_{A \Rightarrow B}$ be a reduction from polynomial optimization problem $A$ to polynomial optimization problem $B$. It is called \textbf{$(d^A, c^A, d^B, c^B)$ pseudo-solution preserving} if for any degree-$d^A$ value-$c^A$ pseudo-solution for any instance $\Phi^A_{n_A}$, there is a degree-$d^B$ value-$c^B$ pseudo-solution for its corresponding instance $\Phi^B_{n_B}$, for any $n_A$.   Here $d^A,c^A,d^B,c^B$ should be thought of as functions of $n_A$.
\end{definition}

It is straightforward to verify that the above three properties are transitive. Thus, it is possible to design a chain of reductions for complicated reductions. 

\begin{proposition}[Transitivity of Properties of Reductions] \label{prop:transitivity}
Let $R_{A \Rightarrow B}$ and $R_{B \Rightarrow C}$ be reductions from optimization problem $A$ to optimization problem $B$ and from optimization problem $B$ to optimization problem $C$ respectively. Let $R_{A \Rightarrow C}$ be the natural composition of $R_{A \Rightarrow B}$ and $R_{B \Rightarrow C}$.  
\begin{itemize}
  \item If $R_{A \Rightarrow B}$ is $(s^B, s^A)$-approximate and $R_{B \Rightarrow C}$ is $(s^C, s^B)$-approximate, then $R_{A \Rightarrow C}$ is $(s^C, s^A)$-approximate. 
  \item If $R_{A \Rightarrow B}$ and $R_{B \Rightarrow C}$ are embedded, then $R_{A \Rightarrow C}$ is embedded. 
  \item If $R_{A \Rightarrow B}$ is $(d^A, c^A, d^B, c^B)$ pseudo-solution preserving and $R_{B \Rightarrow C}$ is $(d^B, c^B$, $d^C, c^C)$ pseudo-solution preserving, then $R_{A \Rightarrow C}$ is $(d^A, c^A, d^C, c^C)$ pseudo-solution preserving.
\end{itemize}
\end{proposition}

\begin{proof}
Immediate from the definitions.
\end{proof}

We are ready to illustrate how reductions help establish integrality gaps. 

\begin{proposition} \label{prop:SoS_reduction}
Let $A, B$ be polynomial optimization problems. Let $R_{A \Rightarrow B}$ be the reduction from $A$ to $B$. Assuming there exists a degree-$d^A$ value-$(c^A,s^A)$ integrality gap for $A$, if $R_{A \Rightarrow B}$ is $(s^B, s^A)$-approximate and $(d^A, c^A, d^B, c^B)$ pseudo-solution preserving, then there exists a degree-$d^B$ value-$(c^B,s^B)$ integrality gap for $B$. 
\end{proposition}

\begin{proof}
It follows directly by definition.
\end{proof}

As a direct consequence, Proposition~\ref{prop:SoS_reduction} suggests that we can make use of reductions to derive SoS lower bounds. The hard part is, however, the design of reductions with approximation and pseudo-solution preserving properties. Here, we describe a simple but useful observation that constructs a pseudo-solution from another under a polynomial map. 
We say that $p:\bbR^n\ra \bbR^m$ is a degree-$d$ polynomial map if $p(x) = (p_1(x),\ldots,p_m(x))$ where each $p_i \in \bbR[x]_d$.

\begin{lemma} \label{lem:poly_map_pseudo_distribution}
Let $A\subset \bbR^n, B \subset \bbR^m$ be algebraic varieties, meaning that 
\begsub{varieties}
A & = \{x \in \bbR^n : g_1(x) = \cdots = g_{n'}(x)=0\} \\
B & = \{x \in \bbR^m : h_1(x) = \cdots = h_{m'}(x)=0\} ,
\endsub
for some polynomials $\{g_i\}, \{h_i\}$. 

Suppose that $p$ is a degree-$d$ polynomial map from $\R^n \ra \R^m$ such that $p(A)\subseteq B$.
Let $\psE_A \in \bbR[x_1,\ldots,x_n]_\ell^*$ be a degree-$\ell$ pseudo-expectation that is compatible with the constraints $g_1,\ldots,g_{n'} \in \bbR[x_1,\ldots,x_n]$.  
Then there exists a degree-$\ell/d$ pseudo-expectation $\psE_B \in \bbR[y_1,\ldots,y_m]_{\ell/d}^*$ that is compatible with the constraints $h_1,\ldots,h_{m'}$.
\end{lemma}

Note that this is essentially the same statement as Fact A.8 in \cite{BKS}, which was stated there without proof.  

\begin{proof}
Assume that $A$ is a discrete set.  This matches our actual application in which $A = \{\pm 1\}^n$ and mostly affects only the notation.  
We will need to introduce the notion of a pseudo-density.  A degree-$l$ pseudo-density on $A$ is a function $\mu_A : A \ra \R$ such that $\sum_{x\in A} \mu_A(x)=1$ and $\sum_{x\in A} \mu_A(x) f(x)^2 \geq 0$ for all $f \in \R[x]_{l/2}$.  The term ``pseudo-'' refers to the fact that $\mu_A(x)$ can be negative.  Any true probability distribution is also a pseudo-density and in the case of $A=\{\pm 1\}^n$, degree-$n$ pseudo-densities are also probability distributions.  In general a degree-$\ell$ pseudo-density $\mu_A$ induces a pseudo-expectation $\psE_A\in \bbR[x]_\ell^*$ with 
\be \psE_A[f] := \sum_{x\in A} \mu_A(x) f(x) ,\label{eq:pseudery}\ee
for all $f\in \bbR[x]_l$.

 To obtain a pseudo-density from a pseudo-expectation we need to solve an underconstrained system of linear equations. This can be done as follows.  Let $e_A :\bbR[x] \ra \bbR^A$ denote the evaluation map on $A$; i.e. $e_A(f)$ is the tuple $(f(x))_{x\in A}$.  Note that $e_A$ is a linear map, and we can also view $\mu_A$ as a linear map from $\R^A\ra \R$.  Given a degree-$\ell$ pseudo-expectation $\psE_A$, \eq{pseudery} can be thought of as constraining $\mu_A$ on the subspace $e_A(\R[x]_\ell)$.  If we write $\R^A  = e_A(\R[x]_\ell) \oplus V$ for some subspace $V$ then we can extend $\mu_A$ to act arbitrarily on $V$.  As long as the action on $\R[x]_\ell$ is the same, this will still meet the definition of a pseudo-distribution.

Now define
\[
   \mu_B(y) = \sum_{x \in p^{-1}(y)} \mu_A(x), \quad\forall y \in B
\]
Since $p(A)\subseteq B$ we have $\sum_{y\in B} \mu_B(y) = \sum_{x\in A}\mu_A(x) = 1$.  And if $f \in \R[y_1,\ldots,y_m]_{\ell/2d}$ then 
\be\sum_{y\in B} \mu_B(y) f(y)^2 = 
\sum_{x\in A} \mu_A(x) f(p(x))^2 \geq 0,\ee
since $\deg(f\circ p)\leq \ell/2$.
Thus $\mu_B$ is a valid pseudo-density.

Finally we can define $\psE_B \in \R[y]_{\ell/d}^*$ by 
\be \psE_B[f] := \sum_{y\in B} \mu_B(y) f(y).\ee
  By the above arguments, $\psE_B[1]=1$ and $\psE_B[f^2]\geq 0$ whenever $\deg f \leq \ell/2d$.  Also, for any $i \in [m']$ and any $q\in \R[y]_{\ell/d - \deg(h_i)}$ we have
\be \psE_B[h_iq] = \sum_{y \in B} \mu_B(y) h_i(y) q(y)
= 0,\ee
since $h_i(B)=0$.  Thus $\psE_B$ is compatible with the constraints $h_1,\ldots,h_{m'}$.
\end{proof}

The previous lemma implies the following corollary, which allows us to obtain perfectly satisfying pseudo-solutions for CSPs via ``local'' reductions.
\begin{proposition}\label{prop:low-degree}
Let $A,B$ be CSPs with a reduction $R_{A\Ra B}$.  Suppose that
\begin{itemize}
\item there exists a map $f:\cP_{n_A}^A \ra \cP_{n_B}^B$, such that if $x_A \in \cP_{n_A}^A$ satisfies all the constraints of an instance of $A$, then $f(x_A) \in \cP_{n_B}^B$ satisfies all the constraints of the corresponding instance of $B$,
  \item each coordinate of $f(x_A)$ depends on at most $\kappa$ coordinates of $x_A$,  and
  \item there exists a degree-$d^A$ pseudo-solution that perfectly satisfies $A$.
    \end{itemize}
    Then there exists a degree-$d^A/\kappa$ pseudo-solution that perfectly satisfies $B$.
\end{proposition}
\begin{proof}
  By a similar procedure to the proof of Proposition~\ref{prop:csp_to_poly}, we can express each coordinate of the mapping $f$ as a polynomial of degree $\kappa$. 
  Now, suppose that the source problem $A$ has a perfectly satisfying degree-$d^A$ pseudo-solution, given by a pseudo-expectation operator $\psE^A[\cdot]$. Then \lemref{poly_map_pseudo_distribution} applied to the degree-$\kappa$ polynomial maps constructed above, yields a degree-$d^A/ \kappa $ perfectly satisfying pseudo-solution for $B$.
\end{proof}

\subsection{ $\txor$ with integrality gap} \label{sec:txor}
In this section, we will introduce the base hard problem underlying
our reductions,
which is the $\txor$ problem first discovered by Grigoriev~\cite{Grigoriev01} and subsequently rediscovered by Schoenebeck~\cite{Schoenebeck08}. 
It is analogous to the proof that 3-SAT is NP-hard, from which other hardness results can be derived by reducing those problems to 3-SAT.
In our framework, $\txor$ can be formulated as follows.

\begin{definition}[$\txor$] \label{def:txor}
$\txor$ is a boolean polynomial optimization problem with the following restriction:
\begin{itemize}
  \item \textbf{Instances}: for any $n$, an instance is parameterized by a formula
    $\Phi_n$ that consists of a set $\mathcal{C}$ of $m=m(n)$ $\txor$ clauses on $n$
    boolean ($\pm 1$) variables.  In other words, we have the constraints $x_i^2=1$ for
    each $i$ and the objective function is  
  \[
     \Phi_n(x) =  \frac{1}{m} \sum_{ (i,j,k) \in \mathcal{C}}
\frac{1 + a_{ijk} x_i x_j x_k }{2}.
  \]
\end{itemize}
Thanks to the $x_i^2=1$ constraints, these terms are equivalent to ones of the form $(1 - (x_ix_jx_k - a_{ijk})^2)/2$.
\end{definition}

Grigoriev's result~\cite{Grigoriev01} (reformulated by Barak~\cite{Barak14}) implies the following integrality gaps. We have a slightly different formulation from~\cite{Barak14} that is slightly stronger but guaranteed by~\cite{Grigoriev01}.

\begin{proposition}[Theorem 3.1 of~\cite{Barak14}, due to Grigoriev]
  For any $\eps>0$, for every $n$ large enough there exists a $\txor$
  instance $\Phi_n$ with $n$ variables and $m=O(n/\eps^2)$ 
  clauses, such that $\OPT(\Phi_n) \leq \frac{1}{2} + \eps$, but there
  exists a degree-$\Omega(n)$ \emph{perfectly satisfying}
  pseudo-solution $\psE$.  In other words there is a
  degree-$\Omega(n)$ value-$(1,\frac 1 2  +\eps)$ integrality gap for $\txor$.  
  \label{prop:grigoriev}
\end{proposition}
  Recall that ``perfectly satisfying'' means that for every clause
  $x_i x_j x_k = a_{ijk}$, it holds that
  $\psE[(x_i x_j x_k - a_{ijk})p(x)] = 0$ for all polynomials $p(x)$
  with degree at most $d- 3$.

\section{Lower bounds on $h_{\Sep}$ and its applications}\label{sec:lb_hsep}
In this section, we will explain how the lower bounds on $h_{\Sep}$
are derived through reductions in the framework introduced in
Section~\ref{sec:lb_framework}.  We will describe the
high-level reduction path here and then explain each reduction in detail in following subsections. 

\[
  \txor  \underset{R_1}{\Longrightarrow} \tofsateq \underset{R_2}{\Longrightarrow} \qmat  \underset{R_3}{\Longrightarrow} h_{\Sep}
\]

There are three reductions $R_1, R_2, R_3$ respectively in the reduction path from $\txor$ to $h_{\Sep}$. The starting point is $\txor$ as we introduced in Section~\ref{sec:txor}. 
We need to define two intermediate problems.
\begin{itemize}
 \item $\tofsateq$ is a boolean polynomial optimization in which each instance is parameterized by a formula $\Phi_n$ that consists of $\tof$ clauses and \textsc{EQ} clauses. 
    \begin{itemize}
      \item  Each $\tof$ clause involves 4 boolean variables $x_i,x_j,x_k,x_l \in \zo$. The clause is satisfied if and only if exactly 2 out of 4 variables $x_i,x_j,x_k,x_l$ are true ($+1$). 
      \item  Each \textsc{EQ} clause involves 2 boolean variable $x_a, x_b \in \zo$. The clause is satisfied if and only if $x_a=x_b$. 
    \end{itemize}
These acceptance conditions correspond to the 0/1-valued predicates
\be \frac{1 + x_ix_jx_kx_l}{2} \qand \frac{1 + x_ax_b}{2}.\ee
If we define $\Phi_n(x)$ to be the fraction of clauses in $\Phi_n$ satisfied by $x$ we can then see that $\Phi_n(x)$ is a degree-4 polynomial in $x_1,\ldots,x_n$.
 \item $\qmat$ is the ``honest prover'' acceptance probability of a QMA(2) protocol
(see \secref{tofsateq_qmat} for details) for $\tofsateq$
This QMA(2) protocol 
\end{itemize}

We note that most of this chain of reductions is implicit in the earlier works of
Aaronson, Beigi, Drucker, Fefferman~and Shor (ABDFS)~\cite{AaronsonBDFS08} and Harrow and
Montanaro~\cite{HM13}, which show reductions from 3SAT to
$h_{\Sep}$. 
The only exception is that we replace the application of PCP theorem, which used to be the first step in reductions and turns out to be a high-degree reduction, by some direct and low-degree construction inspired by part of the proof of the PCP theorem. 

Moreover, our argument requires explicit analysis of the
intermediate steps of the chain to make sure that individual reductions are pseudo-solution preserving (see
Definition~\ref{def:reduction_pseudo_sol}) and thus low-degree.
These explicit analysis of each step will help us further to enable application of the LRS result, which we will elaborate on in Section~\ref{sec:lrs}.


Precise definitions of each problem will appear in each corresponding
subsection. All three reductions will be elaborated on in
Section~\ref{sec:txor_tofsateq} and \ref{sec:tofsateq_qmat}, as well as the SoS hardness result of the $h_{\Sep}$ and 2-to-4 norm problem . We will briefly describe extensions to
other ETH-based hardness results in
Section~\ref{sec:other-bounds}. 

\subsection{From $\txor$ to $\tofsateq$} \label{sec:txor_tofsateq}

In this section we show an explicit reduction from $\txor$ to $\tofsateq$
that preserves the pseudo-solutions and has reasonable approximation parameters. 
The following proposition shows the reduction has reasonable approximation parameters as well as 
some other useful features for later reduction steps.

\begin{proposition} \label{prop:reduction_txor_tofsat}
  For all $m$, $n$, there exists a reduction that  maps a given $\txor$
  instance $\Psi$ with $m$ clauses and $n$
  variables, onto a $\tofsateq$ instance $\Phi$ satisfying the following properties:
  \begin{enumerate}
    \item Every variable in $\Phi$ appears in at most $O(1)$ clauses. \label{item:bounded_degree_csp}
    \item $\Phi$ has $O(n+m)$ variables and $O(m)$ clauses. \label{item:bounded_vars}
  \item If $\Psi$ is perfectly
    satisfiable, then so is $\Phi$. \label{item:completeness}
  \item If at most $1 - \delta$
    fraction of the clauses of $\Psi$ are satisfiable, then at most $1 -
    \Omega(\delta)$ fraction of the clauses of $\Phi$ are
    satisfiable. \label{item:soundness}
  \end{enumerate}
  \label{prop:reduction}
\end{proposition}
\begin{proof}
  We perform the reduction in two steps. First, we reduce the $\txor$ 
  instance to a $\tofsateq$ instance in a manner that preserves
  properties~(\ref{item:bounded_vars})-(\ref{item:soundness}). Next,
  we achieve property~(\ref{item:bounded_degree_csp}) without losing the
  others through an ``expanderizing'' step, similar to the degree
  reduction in Dinur's proof of the PCP theorem. We now describe the two steps
  in turn.

 \vspace{1mm} \noindent \textbf{Step 1}: we show how to transform each $\txor$ clause into
  three $\tof$ clauses, each acting on one of the original $\txor$ 
  variables and two new dummy variables. Altogether we introduce three
  new dummy variables per 3XOR clause. Additionally, in order to break
  the symmetry of $\tofsat$ under parity reversal, we
  introduce a \emph{parity reference bit}, which we denote
  $z$.
Suppose for now that $z = 1$.  
Now, suppose we have a 3XOR clause $c = [ x_i x_j x_k = a_{ijk}]$. First we treat
  the case where $a_{ijk}= 1$. We introduce
  three new variables $y^c_1,y^c_2,y^c_3$, and generate the following three
  $\tof$ clauses:  $\tof(x_i,y^c_2,y^c_3,z)$, $\tof(x_j, y^c_1,y^c_3,z)$, and $\tof(x_k,y^c_1,y^c_2,z)$.
  If we fix an assignment to $x_i, x_j, x_k$, it is easy to see that
  if $x_i x_j x_k = a_{ijk}$ then there exists an assignment to $y^c_1,y^c_2,y^c_3$
  that satisfies all three clauses; otherwise, at most two of the
  three clauses are satisfied for all assignments to $y^c_1, y^c_2, y^c_3$.  In particular if $(x_i,x_j,x_k)=(1,1,1)$ then we set $(y^c_1,y^c_2,y^c_3)=(-1,-1,-1)$ and if $(x_i,x_j,x_k)$ has Hamming weight 1 then we set 
$(y^c_1,y^c_2,y^c_3)=(x_i,x_j,x_k)$.  On the other hand, since each $y^c_i$ appears in two
clauses, multiplying all the clauses yields $(x_iy^c_2y^c_3z) (x_jy^c_1y^c_3z) (x_ky^c_1y^c_2z)=x_ix_jx_kz$.  If this equals $-1$ then not all of the $\tof$ clauses can be satisfied.
If  $a_{ijk} = -1$, then we simply replace $z$ with $\neg z$ in the
  $\tof$ clauses, and the same story holds. 

Applying this
  transformation to all the clauses of the $\txor$ instance yields a
  $\tofsat$ instance with $n + 3m + 1$ variables and $3m$ clauses. If the
  original satisfying fraction was $1$, then the resulting instance
  also has satisfying fraction $1$; otherwise, if the original
  satisfying fraction was $1 - \delta$, the new instance has
  satisfying fraction at most $1 - \delta/3$.

The above analysis holds only when $z=1$.  If we set $z=-1$ then all satisfied $\txor$ clauses become unsatisfied and vice-versa. However, this symmetry already existed in the original $\txor$ formula.  Indeed replacing $x_1,\ldots,x_n$ with $-x_1,\ldots,-x_n$ would have the same effect.  Thus we can assume WLOG that $z=1$.

\vspace{1mm} \noindent \textbf{Step 2}: The resulting $\tofsat$ instance may have some variables that occur in a large number of clauses. Indeed, the parity reference bit occurs
  in all of the clauses. To fix this, we shall apply Lemma~\ref{lem:expanderize} that fixes this issue while keeping all other properties.
\end{proof}

\begin{lemma}[Degree reduction] \label{lem:deg_red}
  There exists a process that maps any instance $G$ of $\tofsat$ to an instance $G'$
  of $\tofsateq$ 
  where 
  \begin{enumerate}
    \item Every variable appears in at most 4 constraints.
    \item If $G$ has $m$ clauses, then $G'$ has $\leq O(m)$ clauses.
    \item If $\OPT(G) = 1$, then $\OPT(G') = 1$.
    \item If $\OPT(G) = 1- \epsilon$, then $\OPT(G') \leq 1 - \eta
      \epsilon$ for constant $\eta$.
  \end{enumerate}
  \label{lem:expanderize}
\end{lemma}
\begin{proof}
  We use the ``expanderization'' process introduced by Papadimitriou
  and Yannakakis~\cite{PY91}. Specifically, we replace every variable
  that occurs in too many clauses by copies, with equality checks
  between them arranged according to a degree-3 expander
  graph. 
\end{proof}

In the following we demonstrate the above reduction also preserves the pseudo-solutions. 

\begin{proposition} \label{prop:witness_tofsateq}
For some constant $0<\delta<1$, there exists a degree-$\Omega(n)$ value-$(1,1-\delta)$ integrality gap for the $\tofsateq$ problem. 
  Moreover, if for any $\tofsat$ clause $\tof(x_i, x_j, x_k, x_\ell)$ in any instance $\Phi_n$, we have for  any polynomial $p(x)$ of
  degree at most $d - 4$, 
  \[ \psE[ p(x) (x_i+ x_j +x_k+ x_\ell)] = 0, \]
  where $\psE$ is from the pseudo-solution of the integrality gap. 
  \label{prop:tofpseudo}
\end{proposition}

To keep the notation simple in the above proposition we ignore the fact that a $\tof$ clause could in general be of the form $\tof((-1)^a x_i, (-1)^b x_j, (-1)^c x_k, (-1)^d x_\ell)$.

\begin{proof}
We start with the degree-$\Omega(n)$ value-$(1,\frac{1}{2} + \eps)$ integrality gap from Proposition~\ref{prop:grigoriev} for $\txor$.
Using the reduction in Proposition~\ref{prop:reduction_txor_tofsat}, we obtain corresponding instances in the $\tofsateq$ problem that have 
true optimum value at most $1-\delta$ for some constant $0<\delta<1$.  

It then suffices to establish pseudo-solutions for these instances in the $\tofsateq$ problem. To do this, we recall the map between satisfying assignments defined in Proposition~\ref{prop:reduction_txor_tofsat}:
\begin{itemize}
  \item Each variable $x \in \zo$ from the original $\txor$ instance is mapped to a variable in the $\tofsateq$ instance with the same assigned value.
  \item The $\tofsateq$ instance has a parity reference bit $z$ that is set to be 1. 
  \item For each $\txor$ clause $c$, we introduce 3 dummy variables $y_1^c, y_2^c, y_3^c$ in the $\tofsateq$ instance. For every satisfying assignment to the clause $c$, there exists an satisfying assignment to the dummy variables that depends only on the assignments of the variables in the clause $c$.
  \item The copies of variables in the expanderization step (Lemma~\ref{lem:deg_red}). 
\end{itemize}
Thus, the hypotheses of Proposition~\ref{prop:low-degree} are satisfied with $\kappa = 3$. Hence, by applying that proposition to the perfectly satisfying pseudo-solution given by Proposition~\ref{prop:grigoriev}, we obtain a degree-$\Omega(n)$ perfectly satisfying pseudo-solution as desired. All in all, this gives us a degree-$\Omega(n)$ value-$(1,1-\delta)$ integrality gap for $\tofsateq$.

\end{proof}

\subsection{From $\tofsateq$ to $\qmat$} \label{sec:tofsateq_qmat}
We start with a description of a QMA(2) protocol for the problem $\tofsateq$, and the formal definition of
$\qmat$. We will show that the fact that our QMA(2) protocol solves
$\tofsateq$ implies a reduction from $\tofsateq$ reduces to $\qmat$.  
Later on, we will use the connection between QMA(2) and $h_{\Sep}$ to reduce
from $\qmat$ to $h_{\Sep}$. 

Our protocol is constructed by starting with a slight modification of the QMA(k)
protocol due to ABDFS~\cite{AaronsonBDFS08} and then converting this QMA(k)
protocol into a QMA(2) protocol using reduction of Harrow and
Montanaro~\cite{HM13}.  The following is a precise statement of the properties
of the resulting QMA(2) protocol.
\begin{proposition} \label{prop:qma2_protocol}
 For any constant $0<\delta<1$ and any constant $\eps>0$, there exists a QMA(2) protocol $P$ for $\tofsateq$ such that 
   \begin{enumerate}
     \item If any instance $\Phi$ of $\tofsateq$ has $\OPT(\Phi)=1$, then protocol P accepts with probability 1 with the following quantum witness $\ket{\psi_x} \ot \ket{\psi_x}$, where \footnote{We use the convention that $\tilde{O}(\cdot)$ hides constants as well as $\polylog(n)$ terms.}  
     \begin{equation} \label{eqn:psi_x}
         \ket{\psi_x}=\left( \frac{1}{\sqrt{n}}\sum_{i=1}^{n}{x_i}
  \ket{i}\right)^{\otimes \tilde{O}(\sqrt{n})}, \forall \text{ satisfiable assignment } x \in \zo^n. 
     \end{equation}
     \item  If any instance $\Phi$ of $\tofsateq$ has $\OPT(\Phi)\leq 1-\delta$, then protocol P accepts with probability at most $\eps$ on any separable quantum witness.
    \end{enumerate}

\end{proposition}
\begin{proof}
We construct such a QMA(2) protocol by composition of a slight modification of the QMA(k) protocol due to \cite{AaronsonBDFS08} and \cite{HM13}. 

 The QMA(k) protocol due to ABDFS~\cite{AaronsonBDFS08} can be used to solve $\tofsat$ when $k=\tilde{O}(\sqrt{n})$. So we just need to show how to modify the protocol to handle equality clauses.  
In the original ABDFS protocol one of the tests was to project onto a random subspace spanned by $\{\ket i, \ket j, \ket k, \ket\ell\}$ with $i,j,k,\ell$ corresponding to some $\tof$ clause.  With probability 1/2 we will replace this with a test that first projects onto the span of $\{\ket i, \ket j\}$ where $i,j$ come from a random equality clause $x_i = x_j$.  Next, we check whether the state is orthogonal to $\ket{i} - \ket{j}$. It is not hard to see the analysis therein works with this slight change, as long as we have the same regularity condition that each variable participates in $O(1)$ equality constraints.
  
Then we can apply the generic reduction from any QMA(k) to QMA(2) from~\cite{HM13}. The final protocol is the composition of the above two
steps. It takes as input a two-party separable state, where each half
of the state consists of $\sqrt{n} \,\polylog(n)$ qubits. These are
grouped into registers of $k = \lceil \log_2(n) \rceil$ qubits each. The protocol consists
of the following four tests, which we describe briefly (for full
descriptions, see the orignal works):
\begin{enumerate}
\item Product test~\cite{HM13}: Think of the proofs from the two provers as
  divided into pieces $A_1,\ldots,A_m$ and $B_1,\ldots,B_m$ respectively, of
  $k$ qubits each, where
  $m = \sqrt{n}\polylog(n)$.  In this test the verifier projects the registers
  $A_iB_i$ onto the bipartite symmetric subspace $\Sym^{2}\bbC^{2^k}$ for each $i$ and rejects if any $A_iB_i$ is found in the antisymmetric subspace.
\item Symmetry test~\cite{AaronsonBDFS08}: In this
  test, the verifier projects $A_1,\ldots,A_m$ onto the $m$-partite symmetric
  subspace $\Sym^{m} \bbC^{2^k}$ and similarly for $B_1,\ldots,B_m$.  (We could also
  combine this with the Product test and simply project $A_1,\ldots,B_m$ onto
  $\Sym^{2m}\bbC^{2^k}$; this would leave completeness the same and would only improve the
  soundness.  The only reason to consider the Product and Symmetry tests separately is so
  that we can use the analyses of \cite{AaronsonBDFS08} and \cite{HM13} as black boxes.)
\item Uniformity test~\cite{AaronsonBDFS08}: The verifier chooses a maximal matching
  $\mathcal{M}$ on the set $\{1, \dots, n\}$ at random, and constructs an
  orthonormal basis containing the vectors 
  \[ \frac{1}{\sqrt{2}}(\ket{u} \pm \ket{v}) \]
  for every edge $(u,v) \in \mathcal{M}$. The verifier then measures each $A_i$
  and each $B_j$ in this basis, and rejects if for some $(u, v) \in
  \mathcal{M}$, the outcomes $\frac{1}{\sqrt{2}}(\ket{u} \pm \ket{v})$ are both
  obtained on different subsystems.
\item Satisfiability test~\cite{AaronsonBDFS08}, modified to handle equality
  clauses as described above: The verifier divides the
  clauses into a constant number of blocks such within each block, no two
  clauses share a variable. For each piece of the proof $A_i$ or $B_j$, the
  verifier chooses a block $B$ at random, and applies a POVM whose elements are
  the projectors $\Pi_{C}$ onto the span $\mathrm{span}(\{\ket{i}: i \in C \})$
  of computational basis states corresponding to the variables in each clause
  $C$ in the block $B$. This is completed into a POVM by including the POVM element $M_{\text{default}}  = I - \sum_{(i,j,k,\ell) \in
    B} \Pi_{i,j,k,\ell}$. If the measurement returns outcome $C$ corresponding
  to some clause $C$, then the verifier checks whether this clause is satisfied
  by the witness. If the clause $C$ is a $\tofsat$ clause on variables
  $i,j,k,\ell$, this is done by checking whether the post-measurement state is orthogonal to
  $\frac{1}{2}(\ket{i} + \ket{j} + \ket{k} + \ket{\ell})$, and if it is an
  equality clause on variables $i, j$, we check whether the post-measurement state is orthogonal
  to $\frac{1}{\sqrt{2}} (\ket{i} - \ket{j})$. If ever any clause is discovered
    to be violated, the verifier rejects; otherwise, he accepts.
\end{enumerate}
By composing the completeness and soundness results
of~\cite{AaronsonBDFS08} and~\cite{HM13}, we obtain that the protocol
given above has completeness $1$ and constant soundness $s < 1$.

Finally, to achieve arbitrarily small constant soundness, we perform
an amplification procedure. Since the measurement operator corresponding
to an accepting outcome is separable, by Lemma 7 of~\cite{HM13}, we can amplify
the soundness of the protocol by performing parallel
repetition: if we start with soundness $s$ and repeat $\ell$ times in
parallel, the new soundness is at most $s^\ell$.

\end{proof}

We are ready to formally define $\qmat$. 

\begin{definition} \label{def:qmat}
$\qmat$ is a boolean polynomial optimization problem 
with objective function $\Psi: \zo^n \rightarrow [0,1]$ defined as follows.
Let $\Phi$ be an instance of $\tofsateq$
and $P$ the corresponding protocol from Proposition~\ref{prop:qma2_protocol}.
Then
  \begin{eqnarray*}
     \Psi(x) &  :=  & \Prob [ P \text{ accepts on input } \Phi \text{ and witness } \ket{\psi_x}\ot \ket{\psi_x}], \\
                &  = & \tr(M_n^\Phi \ket{\psi_x}\bra{\psi_x} \ot \ket{\psi_x}\bra{\psi_x}),       
  \end{eqnarray*}
  where
$x \in \zo^n$, 
 $M_n^\Phi$ is the POVM corresponding to the acceptance in protocol $P$ on input $\Phi$ and $\ket{\psi_x}$ is defined in (\ref{eqn:psi_x}). 
\end{definition}
It is important to note that the objective function $\Psi(x)$ is a function of
the boolean variables $x$, \emph{not} of the components of the witness state
$\ket{\psi_x}$, and hence $\qmat$ is indeed a boolean optimization problem. For
this to work, we rely on the fact that $\Psi(x)$ is a
degree-$\tilde{O}(\sqrt{n})$ polynomial in the variables $x$, which follows from
the fact that the components of the witness state $\ket{\psi_x}$ are themselves
degree-$\tilde{O}(\sqrt{n})$ polynomials of $x$.

\subsubsection*{Integrality gap for $\qmat$}
We are ready to establish the integrality gap for $\qmat$ from the integrality gap for $\tofsateq$ using the above reduction, i.e., each instance $\Phi \in \tofsateq$ is reduced to an instance $\Psi \in \qmat$ as in the Definition~\ref{def:qmat}. First, we note a useful fact that allows us to interpret pseudoexpectations as quantum states.

\begin{definition} \label{def:conn}
 Let $\psE_x[\cdot]$ be the pseudoexpectation operator corresponding to
 a degree-$2d$ pseudodistribution over $n$ variables $\{x_1, \dots,
 x_n\}$ satisfying the constraints $x_i^2 - 1 = 0$ for all $i$ from
 $1$ to $n$. (This means $\psE[(x_i^2 - 1) p(x)] = 0$ for all $i$ and for
 all polynomials $p(x)$ with degree at most $2d - 2$.) Then we define
 the associated mixed quantum state $\rho_x$ to be the following
 state over $d \log n$ qubits:
 \[ \rho_{x}  = \frac{1}{n^d} \sum_{\substack{i_1 i_2 \dots i_{d}\\ j_1 j_2 \dots j_d}}
 \psE_x[x_{i_1} \dots x_{i_d} x_{j_1} \dots x_{j_d}] \ket{i_1 \dots
   i_d} \bra{j_1 \dots j_d}.\]
\end{definition}
If $\psE_x[\cdot]$ is a valid pseudoexpectation operator, then
$\rho_x$ is PSD and has trace $1$, so it is indeed a quantum mixed
state. Moreover, by the $x_i^2 - 1 = 0$ constraint, we find that the
reduced density matrices correspond to the low-degree moments of the
pseudodistribution.

\begin{proposition} \label{prop:witness_qmat}
For any constant $\eps>0$, there exists a degree-$\Omega(n)$ value-$(1,\eps)$ integrality gap for $\qmat$. 
\end{proposition}

\begin{proof}
For any $n$, let $\Phi_n \in \tofsateq$ be the instance from the degree-$\Omega(n)$ value-$(1,1-\delta)$ integrality gap for $\tofsateq$ from Proposition~\ref{prop:witness_tofsateq} and $\mu_n$ be the corresponding pseudo-solution ($\delta$ is the constant therein). Let $\Psi_n$ be its reduced instance of $\qmat$. 
For any constant $\eps>0$, 
by Property (2) of Proposition~\ref{prop:qma2_protocol} and the fact $\OPT(\Phi_n) \leq 1-\delta$, we have $\OPT(\Psi_n) \leq \eps$.
Then it suffices to establish a $(\Omega(n),1)$ pseudo-solution for $\Psi_n$. 

To that end, we claim that $\mu_n$ is also a degree-$\Omega(n)$ value-1 pseudo-solution for $\Psi_n$.  Note that the feasible sets for $\Phi_n$ and $\Psi_n$ are the same. It suffices to show $\psE_{x \sim \mu} [\Psi_n(x)]=1$. Observe that by linearity of $\psE[\cdot]$ and $\tr(\cdot)$, we have
\be
  \psE_{x \sim \mu} [\Psi_n(x)] = \psE_{x \sim \mu} \left [\tr(M^{\Psi_n} \ket{\psi_x}\bra{\psi_x} \ot \ket{\psi_x}\bra{\psi_x})\right]= \tr (M^{\Psi_n}\tilde\rho),
\label{eq:MPsi-exp}\ee
where we have defined $\tilde{\rho}$ as 
\begin{equation} \label{eqn:trho} \tilde{\rho}=\psE_{x \sim \mu} \left [ \ket{\psi_x}\bra{\psi_x} \ot \ket{\psi_x}\bra{\psi_x} \right ] . \end{equation} 
We note that this is precisely the state obtained by applying a partial trace to all but $\tilde{O}(\sqrt{n})$ subsystems of the state $\rho_x$ defined above.

Now, we need to calculate the expectation value of $M^{\Psi_n}$, the POVM element
corresponding to the ``yes'' outcome of the protocol. Recall from \propref{qma2_protocol}
that our protocol is obtained by parallel repetition of the protocol of~\cite{HM13}, where
the number of repetitions is constant.  (More precisely the number of repetitions is
$O(\log 1/\eps)$ but we neglect this dependence since we take $\eps$ to be a constant and
we allow the constant in the  $\Omega(n)$ degree to depend on $\eps$.) Thus,
$M^{\Psi_n}$ is a linear combination of tensor products of a constant number of  terms, each of
which implements a randomly chosen test on one of the registers of the witness state. The
complementary POVM element $1 - M^{\Psi_n}$ consists of a linear combination of tensor
products, where each product contains at least one ``no'' outcome of a test. To show that
the state $\tilde{\rho}$ passes $M^{\Psi_n}$ with certainty, it suffices to show that the
expectation value of any such term is $0$. Below, we verify this for each test. 

  \begin{enumerate}
  \item {\bf Symmetry and Product tests:} These tests consist of
    applying the swap test to various pairs of registers in the
    state. Since $\tilde{\rho}$ is fully symmetric under any permutation of
    the indices, we pass these tests with certainty, i.e. $\Tr[(M^{\text{``no'',
        symmetry test}} \otimes M^{\text{rest}}) \tilde{\rho}] = 0$. 
  \item {\bf Uniformity test:} 
  Recall that in the uniformity test, Arthur chooses a matching
  $\mathcal{M}$ on $[n]$, and the measures each subsystem in an
  orthonormal basis containing 
  \[ \ket{\pm}_{ij} \equiv \frac{1}{\sqrt{2}}(\ket{i} \pm \ket{j}) \]
  for every $(i,j) \in \mathcal{M}$. The test fails if for some
  $(i,j)$, outcomes of different subsystems are different. 
  We claim this won't happen with $\tilde{\rho}$. Without loss of generality, let the first two subsystems have different outcomes. 
  The probability for this to happen is given by 
  \begin{align*}
    \Prob[\text{Uniformity test failure}] &= \Tr[\tilde{\rho} \ket{+}_{ij} \bra{+}_{ij} \otimes \ket{-}_{ij}
        \bra{-}_{ij} \otimes M^{\text{rest}} ] \\
      &\propto \psE_{x \sim \mu}[ (x_i + x_j)^2(x_i - x_j)^2 q(x) ], \text{ for some polynomial } q(x) \\
      &= \psE_{x \sim \mu}[(x_i^2 + x_j^2 + 2x_i x_j)(x_i^2 + x_j^2 - 2x_i x_j)
        q(x)] \\
      &= \psE_{x \sim \mu}[(2 + 2 x_i x_j)(2 - 2x_ix_j)q(x)] = \psE_{x \sim \mu}[4(1 - x_i^2 x_j^2) q(x)] = 0.
  \end{align*}
  In the above calculation, we used \eq{boolean-psE}
repeatedly to simplify the terms. Also note that since there are
  $\tilde{O}(\sqrt{n})$ registers in the witness state, the degree of
  $q(x)$ is $\tilde{O}(\sqrt{n})$, which is less than the degree of
  the pseudoexpectation $\Omega(n)$.
  \item {\bf Satisfiability test: } In the satisfiability test, we
    choose a set of clauses to
  measure that have no variables in common with each other. Now, we
  perform the following procedure on the witness: first perform a measurement to project the witness
  into the subspace $\mathrm{Span}(\{ \ket{i} : i \in C \})$ spanned
  by the variables occurring in a clause $C$. 
  
     If we end up in the subspace associated with $C = \tof(x_{i_1}, x_{i_2}, x_{i_3},
  x_{i_4})$, then we perform another projective measurement to check
  that the state is orthogonal to 
  \[ \ket{C}=\frac{1}{2} ( \ket{i_1} + \ket{i_2 } +  \ket{i_3} +
 \ket{i_4}). \]
  Let $\Pi_C = \ket{C}\bra{C}$.   
    Let $m$ be the number of copies of the witness, and suppose that the
  first stage of this test projects us onto clauses $C_1, \dots ,	
  C_m$. Then the probability of the second stage passing is
  \begin{align*}
    \Prob[\text{Success}] &= \Tr\left[ (I - \Pi_{C_1}) \otimes (I - \Pi_{C_2}) \otimes
        \dots \otimes (I - \Pi_{C_m}) \otimes M^{\text{rest}} \tilde{\rho}\right] \\
    &\propto \psE_{x \sim \mu}\left[ (1 - \frac{1}{4} (\sum_{x \in C_1} x)^2)
      \dots (1 - \frac{1}{4} (\sum_{x \in C_m} x)^2) \dots\right].
  \end{align*}
  Now, we know that the pseudo-solution $\mu$ has degree $\Omega(n)$ and satisfies all the $\tof$ constraints.
  In particular for every clause $\psE_x[(\sum_{x \in C_1} x) q(x)]
  = 0$ for all polynomials $q(x)$ with degree $o(n)$. But in the
  expression above, we have a product of $m = \tilde{O}(\sqrt{n})$ terms, each
  of degree $2$, so every term containing a factor of $\sum_{x
    \in C} x$ will vanish under the pseudo-expectation. This leaves us
  with $\Prob[\text{Success}]  = 1$ as desired.
  
  Similarly, if we end up in the subspace associated with $C=\textsc{EQ}(x_{i_1},x_{i_2})$. We do the same thing except now we choose
  \[
    \ket{C} =\frac{1}{\sqrt{2}} ( \ket{i_1} -\ket{i_2}). 
  \]
  The analysis is analogous to the above and we end up with $\Prob[\text{Success}]  = 1$ as desired.
\end{enumerate}
Note that it is crucial that even in the parallel-repeated protocol, the number of subsystems is $\tilde{O}(\sqrt{n})$. This means that all the tests in the protocol translate to polynomials of degree at most $\tilde{O}(\sqrt{n})$ under the pseudoexpectation. Since the pseudoexpectation is valid up to degree $\Omega(n)$, this means that the tests cannot tell that $\tilde{\rho}$ is not an honest witness.
\end{proof}

We can understand Proposition~\ref{prop:witness_qmat} as an
explicit lower bound on the Doherty-Parrilo-Spedalieri~\cite{DPS03}
hierarchy for $h_{\Sep}$.
\begin{corollary} \label{cor:DPS}
  For any constant $\eps$, there exists a family of measurements $M_d$
  acting on a bipartite Hilbert space with local dimension $d$, such
  that $h_{\Sep}(M) \leq \eps$, but the $k^{\text{th}}$-level of DPS estimates this value to be 1, i.e., $\mathrm{DPS}_k(M) = 1$ for $k
  \leq o(\log d / \polylog\log d)$. 
\end{corollary}
\begin{proof}
  We take $M_d$ to be $M^\Psi$ from the QMA(2) protocol, and
  $\tilde{\rho}$ from (\ref{eqn:trho}).
  The state $\tilde{\rho}$ arises as the reduced density matrix of a
  fully symmetric state $\rho$ on $\Omega(n)$ registers and is thus a
  $\tilde{O}(\sqrt{n})$-extendible state. Moreover, since $\rho$ is
  invariant under \emph{all} permutations of indices (including those
  that exchange row and column indices), it is \emph{a fortiori}
  invariant under partial transposes, and hence PPT. Thus, $\rho$ lies
  within the set of states explored by DPS at level $k =
  \tilde{O}(\sqrt{n})$. The statement follows because $d=2^{\tilde{O}(\sqrt{n})}$. 
\end{proof}

In \cite{BHKSZ12} it was shown that computing the $2\ra 4$ norm was a
special case of computing $h_{\Sep}$ and that in turn there was an
approximation-preserving reduction from $h_{\Sep}$ to the $2\ra 4$
norm.  Examining that construction, we see that it is $O(1)$-degree
and this lets us immediately obtain the following bound.

\begin{corollary}\label{cor:24}
The SoS relaxation needs at least $\Omega(\log(d) /
  \polylog\log(d))$ levels to approximate $\|A\|_{2\ra 4}$ up to
multiplicative error of $C = O(1)$.
\end{corollary}

\subsection{DPS lower bounds from other protocols}\label{sec:other-bounds}

There are two other ETH-based hardness results that we can also make
unconditional, but we will only sketch the proof here. Both of these
results are obtained by an argument similar to the proof of \corref{DPS}, but with a different choice of Constraint Satisfaction Problem (CSP) in
place of $\tofsat$, and a correspondingly different reduction from
3XOR and choice of $\text{QMA}(2)$ or
$\text{QMA}(k)$ protocol. 

The first result is an SDP hardness result for $(1,1-\tilde O(1/n))$
approximations to $h_{\Sep(n,n)}$, using a protocol
of~\cite{GNN12}. 

\begin{theorem}\label{thm:sos-lb-gnn}
Let $\qmat$ be the problem of maximizing the acceptance probability of
the protocol of~\cite{GNN12} over honest strategies (to be described
below). Then for every $n$, there exists an instance of the problem
$\qmat$ on $\poly(n)$ variables with true optimum value $\leq 1 -
\tilde{\Omega}(1/n)$. Moreover, there exists a pseudosolution that
achieves value $1$ on this problem. As a consequence, we obtain a
family of measurements $M_d$ with $h_{\Sep}(M) \leq 1 -
\tilde{\Omega}(1/n)$, but for which $\mathrm{DPS}_k(M) = 1$ for $k
\leq o(n)$.
\end{theorem}

\begin{proof}
The  protocol of~\cite{GNN12} solves an NP-hard
graph coloring problem in $\text{QMA}(2)$ with completeness $1$ and soundness
$1 - \tilde{\Omega}(1/n)$. Schematically, the proof of the hardness result
is:

\[ \txor(n) \Longrightarrow \textsc{Graph-3-coloring}(n) \Longrightarrow
\qmat \Longrightarrow h_{\Sep(n, n)}, \]
where $\textsc{Graph-3-coloring}(n)$ is the problem of deciding
whether a graph of $n$ vertices is 3-colorable, and $\qmat$ is defined
as before but with reference to the honest witnesses of the protocol
of~\cite{GNN12}. 
 To achieve hardness for $h_{\Sep}$, we need to show
that the reductions represented by the first two arrows of the diagram is
pseudosolution-preserving, and that the last arrow is an approximate
embedding reduction. The first arrow is a standard construction (see
for instance proposition 2.27 of the textbook by Goldreich~\cite{Goldreich08}), quite
similar to the gadget reduction for $\tofsat$ we considered earlier.
It is straightforward to verify that it satisfies the hypothesis
of proposition~\ref{prop:low-degree} with a constant value of $\kappa$; hence, it is pseudosolution
preserving. 

For the second arrow, we use a strategy similar
to the proof of Proposition~\ref{prop:witness_qmat}, arguing that a
pseudosolution to $\textsc{Graph-3-coloring}(n)$ can be turned into a
dishonest quantum witness state $\ket{\psi}$, and that each test of
the $\text{QMA}(2)$ protocol evaluates a low-degree polynomial on the
coefficients of $\ket{\psi}$, and thus passes with certainty. One
difference from the previous case is that
$\textsc{Graph-3-coloring}(n)$ is not strictly speaking a Boolean
problem. We remedy this by considering colorings that are
\emph{induced} through the reduction from an underlying assignment $x$
to the variables 3XOR instance. For every vertex $i$ and color $c$,
let $g_{c,i}(x)$ be equal to $1$ if vertex $i$ is colored with $c$ in
this induced coloring, and $0$ otherwise. Then, due to the locality of the gadgets in
the reduction from 3XOR, the functions $g_{c,i}(x)$ are
constant-degree polynomials in the variables $x$. Next, we need to
verify that the connection between pseudosolutions and dishonest
quantum witness states still holds in this protocol. Indeed, the
honest witness state in this protocol for a given assignment $x$ is
$\ket{\Psi_x}\bra{\Psi_x} = \left(\ketbra{\psi_x}\right)^{\otimes 2}$, where
\[ \ket{\psi_x} = \sum_{i,c} g_{c,i}(x) \ket{i} \ket{c}.\]
Here the index $i$ runs over the vertices of the graph. Following our
strategy in~\ref{prop:witness_qmat}, we choose the following dishonest
quantum witness state
\[ \rho = \psE_x[ \ketbra{\psi_x} \otimes \ketbra{\psi_x}]. \]
That this is indeed a valid quantum state follows from the properties
of the pseudoexpectation operator, and from
\lemref{poly_map_pseudo_distribution} (the functions $g_{c,i}(x)$
play the role of the polynomial map $p$ in the lemma). 
This gives us a degree-$\Omega(n)$ value $1$ pseudo-solution for $\qmat$. Finally,
the reduction in the last arrow is a $(1 - \tilde{O}(1/n), 1 -
\eta)$-approximate embedding reduction for some constant $\eta$, by
the soundness of the protocol. This yields a degree-$\Omega(n)$ lower
bound for SoS approximations $h_{\Sep(n,n)}$ achieving approximation
factor $(1, 1- \tilde{O}(1/n))$.
\end{proof}

The second result applies for multipartite separability. 
To obtain it,
we replace the protocol of~\cite{HM13} with that \cite{CD10}, which
is a $\text{QMA}(O(\sqrt{n}))$ protocol for 3-SAT with
completness $1 - \exp(-\Omega(\sqrt{n}))$ and soundness $1 -
\Omega(1)$, and which only performs Bell measurements
(i.e.~each party measures individually and then the outcomes are
classically processed). We use this to prove hardness for the problem
$h_{\Sep^{O(\sqrt{n})}(n)}(M)$, where $\Sep^k(n)$ means $k$-partite separable
states, and $M$ is restricted to be a bell measurement. The schematic diagram for this case is
\[ \txor(n) \implies \textsc{Graph-3-coloring}(n) \implies \qmat
\implies h_{\Sep^{O(\sqrt{n})}(n)}(M). \]
The arguments are very similar to those in the previous result; when
we work out the parameters, we obtain that SoS needs at least $\Omega(n)$
levels to achieves a $(1-\exp(-\Omega(\sqrt{n})),1 - \eta)$ approximation to
$h_{\Sep^{O(\sqrt{n})}(n)}(M)$ for general Bell measurements $M$, where
$\eta$ is an appropriately chosen constant.

\begin{theorem}\label{thm:sos-lb-multipartite}
  For a sufficiently small constant $\eta$, there exists a family of
  measurements $M_n$ over $O(\sqrt{n})$-partite states with local
  dimension $n$, such that $h_{\Sep^{O(\sqrt{n})}(n)}(M) \leq 1 -
  \eta$, but $\mathrm{DPS}_k(M) \geq 1- \exp(-\Omega(\sqrt{n}))$ for
  all $k \leq o(n)$.
\end{theorem}

%
%
%

\section{\SDPEF{} Lower Bounds for $h_{\Sep}$}\label{sec:lrs}
In this section, we will leverage our SoS lower bounds to prove a
lower bound on the size of \emph{any} SDP relaxation approximating
$h_{\Sep}$, provided the relaxation is of a certain type called an extended formulation. To that end, we make use of a recent result of Lee, Raghavendra,
and Steurer that relates extension complexity to SoS degree. We start by reviewing the
techniques used by them.

\subsection{\SDPEF{} lower bounds from LRS} \label{sec:LRS}
In this section we illustrate the recent celebrated result of Lee, Raghavendra,
and Steurer (LRS)~\cite{LRS15} lower-bounding the size of \SDPEF{}s of boolean polynomial optimization problems in terms of the sum-of-squares degree.
We will restate their main result (in a slightly more general form) in
our current framework. Note that the LRS result plays a crucial role
in extending our results on SoS hierarchies to more general SDP relaxations.
Our contribution to this topic can be viewed more broadly as developing
techniques to apply LRS to general optimization problems.

To that end, we first define the notion of \SDPEF{}s as follows. 

\begin{definition}[SDP] \label{def:sdp}
A semidefinite program (SDP) $A$ is an optimization problem, parametrized by $n \in \mathbb{N}$,
with the following restrictions.
\begin{itemize}
 \item The feasible set $\cP^A_n$ is a spectrahedron contained in $L_\geq(\R^{r})$, where $r=r(n)$ is called the size of this SDP and $L_\geq(V)$ denotes the set of positive-semidefinite matrices.   By spectrahedron we mean simply a space of the form $W\cap L_\geq(\R^{r})$ for $W$ an affine subspace of $L(\R^{r})$.
 \item Any instance $\Phi^A_n$ is an affine function from $L(\mathbb{R}^{r})$ to $[0,1]$. 
\end{itemize}
\end{definition}

\begin{definition}[\SDPEF{}] \label{def:sdp_relax}
For any optimization problem $A$, an SDP $B$ is called a
\emph{$(c,s)$-approximate extended formulation} of $A$ if there exists an
\emph{embedded} reduction $R_{A \Rightarrow B}$ that is \emph{$(s^B = c,s^A =s)$-approximate}.
\end{definition}
We note that the \SDPEF{} defined above is a more stringent
concept than the conventional notion of an SDP \emph{relaxation}, both because
of the embedding property and because we require that the constraints do not
depend on the objective function.  The intuition behind this definition is that one can decide whether a given
instance of problem $A$ has value $\geq c$ (the ``yes'' case) or $< s$ (the
``no'' case) by solving the SDP corresponding to
$(c,s)$-approximate extended formulation of the problem $A$; if the objective
value of the SDP is at least $c$, one concludes that one is in the ``yes'' case, and
otherwise, that one is in in the ``no'' case.

It is not hard to see that any \SDPEF{} for $B$ is also a \SDPEF{} for $A$ if there is an embedded reduction from $A$ to $B$ with matching approximation parameters. Precisely, 

\begin{proposition} \label{prop:sdp_extend}
Let a SDP $C$ be a $(s^C, s^B)$-approximate \SDPEF{} of an optimization problem $B$. If there is an embedded reduction $R_{A \Rightarrow B}$ that is $(s^B, s^A)$-approximate, then $C$ is a $(s^C, s^A)$-approximate \SDPEF{} of the optimization problem $A$.
\end{proposition}

\begin{proof}
This claim follows by definition and the transitivity of approximating and embedding properties of reductions in Proposition~\ref{prop:transitivity}.
\end{proof}

The above definitions allow us to translate lower bounds on SDPs for one problem to another. To obtain such lower bounds in the first place, we use a technique developed by~\cite{LRS15}, which bounds the positive semidefinite rank of a particular ``pattern matrix.'' In particular, suppose we would like to give an SDP lower bound on $(c,s)$-approximations for an optimization problem. The first observation that LRS make is that to achieve this, it suffices to lower bound the psd rank of the matrix $M^n_{(c,s)}$, whose rows are indexed by instances $\Phi$ whose true optimum value is $\leq s$, and whose columns are indexed by feasible points $x \in \{0,1\}^n$. The value of an entry of this matrix is given by $M^n_{(c,s)}(\Phi, x) = c - \Phi(x)$. Note that all entries of this matrix are nonnegative by construction, so the psd rank is well defined. The second key observation of LRS is that there is a relation between SoS degree of the function $c - \Phi(x)$, and the psd rank of a \emph{different} matrix
\[ M^n_\Phi : [n]^{m} \times \{0,1\}^n \mapsto \mathbb{R}_{\geq 0}, M^n_\Phi(S,x) = c - \Phi(x_S). \]
Here the rows are indexed by subsets $S$ of size $m$, and the notation $x_S$ means the values of $x$ on the coordinates in the subset $S$. The key technical lemma of LRS is the following:
\begin{lemma}[Theorem 3.8 of~\cite{LRS15}]\label{lem:LRS}
  Suppose $\Phi$ is an instance of an optimization problem over $m$ variables, and $\mathrm{deg_{SoS}}(c - \Phi(x)) \geq d$. Then for $n \geq  m^{d/4}$, $\mathrm{rk_{psd}}(M^n_\Phi) \geq \Omega(m^{d^2/8})$. 
\end{lemma}
To relate this to the original problem, we need to show that $M^n_\Phi$ is a \emph{submatrix} of $M^n_{(c,s)}$.

While this approach is so far the most successful route to general SDP
lower bounds, it has two limitations. First, the requirement that $n \geq
m^{d/4}$ implies that the lower bound on psd rank (and hence SDP size)
obtained will never be better than \emph{quasi-polynomial} in $n$. This
requirement seems essential for the random restriction analysis which
is central to the proof of LRS. This means that the bounds obtained
via this method can be much looser than the SoS lower bounds they are
based on. A second limitation appears when we try to use the technique
for settings \emph{other} than CSPs. Essentially, the problem is that
we need to interpret an instance of the problem on $m$ variables as an
instance on $n \gg m$ variables, in order for the matrix $M^n_\Phi$ to be
a submatrix of $M^n_{(c,s)}$. This is straightforward in the case of
CSPs, but not for other problems. For instance, the problems we will
consider in this work arise from particular quantum proof protocols;
these protocols involve states that are superpositions over all of the
variables in the problem, and as a result break down when only a small
number of variables enter into the objective function. As a result of
these limitations, we will only be able to obtain SDP lower bounds in
\emph{some} of the cases where we have SoS lower bounds, and even in
those cases, our parameters will be worse than those of the SoS
results. We consider it a major open problem to improve on these
techniques and prove tighter SDP lower bounds. 

\subsection{Applying LRS to $h_{\Sep}$}
To apply the LRS techniques to $h_{\Sep}$, we need to re-examine each reduction in our chain of 
reductions to ensure the embedding property. First, we need to investigate in more detail the 3-coloring proof system of~\cite{GNN12}, which consists of the following steps:

\begin{enumerate}
\item The verifier receives a product state $\ket{\psi_A}\otimes \ket{\psi_B}$ from the two provers. Each prover's state consists of one register of $\log(n)$ qubits, holding an index from $1$ to $n$, and a second register consisting of a single qutrit, whose three states correspond to the three possible colors in the graph.
\item The verifier performs one of the following tests:
  \begin{itemize}
  \item \textbf{Uniformity test:} For each proof state, the verifier performs a quantum Fourier transform on the color register, and measures it. If he obtains $0$, then he performs an inverse Fourier transform on the index register, and measures it. He accepts if he measures $0$ and rejects otherwise.
  \item \textbf{Satisfiability test:} The verifier measures both proof states in the computational basis, obtaining two tuples $(i, c_i)$ and $(j, c_j)$ of indices and colors. If there is no edge between $i$ and $j$ in the graph, then the verifier accepts. If there is an edge, then the verifier accepts if $c_i \neq c_j$ and rejects otherwise. 
  \end{itemize}
\end{enumerate}
For this protocol, the honest witness states are those of the form
\[ \ket{\Psi} = \left( \frac{1}{\sqrt{n}} \sum_i \ket{i}\otimes \ket{c_i} \right)^{\otimes 2}. \]

When we try to apply LRS directly to the problem $\qmat$ for this
protocol, we run into several obstacles:
\begin{enumerate}
\item LRS requires that the feasible set of the optimization problem
  be the entire Boolean cube $\{0,1\}^n$. This means that we cannot
  take the problem $\qmat$ to be an optimization over all colorings,
  but rather we must restrict ourselves to colorings \emph{induced} by
  Boolean assignments to the variables of an underlying $\txor$
  instance, as we did in the proof of~\thmref{sos-lb-gnn}.
\item The embedding property of LRS means that the SDP feasible point
  corresponding to an assignment $x$ must be \emph{independent} of the
  objective function, i.e. the choice of graph. However, in the
  construction given in the proof of~\thmref{sos-lb-gnn}, the induced 
  coloring depends on the graph.
\end{enumerate}
To address these issues, we make some tweaks to the protocol. First,
for every input size $n$, we choose a \emph{universal} graph $G_n$,
which is induced by a $\txor$ instance with a \emph{complete}
constraint graph. For every assignment $x$ to the $\txor$ variables,
we let the induced coloring $c$ be the coloring induced on the
universal graph $G_n$. It is not hard to see that, in the standard
gadget reduction, the graph obtained will be a subgraph of this
universal graph $G_n$, and the induced coloring will match the
universal induced coloring. 

Now, using this modified protocol, we can prove our main
result. First, we state the soundness property of the protocol in a
form that will be useful to us.
\begin{lemma}[Soundness analysis of~\cite{GNN12}]\label{lem:gnn-soundness}
  There exist a constant $\eta < \frac{1}{2}$ such that if $\Phi$ is a
  $\txor$ instance on $n$ variables with $O(n)$ clauses and value at most $1 -
  \eta$\footnote{In~\cite{GNN12}, these properties are obtained by applying
    Dinur's PCP theorem.}, then the LNN
  protocol accepts with probability at most $1 - \Omega(1/(n \polylog
  n))$. Moreover, for any unsatisfiable $\Phi$ on $n$ variables, then the acceptance
  probability of LNN is at most $1 -\Omega(1/n^2)$.
\end{lemma}
\begin{theorem}\label{thm:lrs-gnn}
Any \SDPEF{}
achieving a $(1-\eps(d), 1-\delta(d))$-approximation to
$h_{\Sep(d,d)}$ where $\delta(d) = O(1/d^2)$ and $\eps(d) < \delta(d)$
has size at least $d^{\log d / \poly\log\log d}$.
\end{theorem}
\begin{proof}
  The proof follows the strategy outlined in \secref{LRS}. First we
  will rule out efficient \SDPEF{}s to $\qmat$, and then show that this
  implies the nonexistence of \SDPEF{}s for $h_{\Sep}$ as well.

  To do this, we use the pattern matrix technique of LRS. 
  Let $\Phi_m(x): \{0,1\}^m \to \mathbb{R}_{\geq 0}$ be the objective
  function of $\qmat$ on instances of size $m$, induced by the hard
  3XOR instance of Grigoriev. Then we know that $\max_{x} \Phi_m(x) \leq 1
  - 1/(n \polylog n)$, and $\mathrm{deg_{SoS}}(1
  -  \Phi_m(c)) \geq \Omega(m)$, i.e. there is a degree-$\Omega(m)$
  pseudodistribution under which $\psE[(1 - \Phi(c))] = 0$. 
  Now, by \lemref{LRS}, this implies that the matrix
  $M^n_{\Phi_m}(S, x) = 1 - \eps(n) - \Phi_m(x_S)$ has PSD rank at least $n^{\Omega(m)}$ for
  $n = m^{\Omega(m)}$. We would like this to be a submatrix of
  $M_{c,s}(\Phi_n, x) = c - \Phi_n(x)$ for some choice of $c,s$, where the instances $\Phi_n$
  are now over $n$ variables. However, this would require that we be able to ``simulate''
  the action of the protocol on $m$ variables using larger instances on $n$ variables, and
  achieve exactly the same objective value. Since the LNN protocol involves sampling
  variables from witness states in uniform superposition, it is not obvious how to do
  this---if we run the protocol over $n$ variables, then most of the samples will lie
  outside any particular subset of size $m$ of the variables.

  In order to avoid this obstacle, we define a new function
  $f'(x):\{0,1\}^m \to \mathbb{R}_{\geq 0}$, which is equal to the success probability of
  the LNN protocol when the witness state is $\ket{\Psi_{x'}}$ where $x' \in \{0,1\}^n$
  agrees with $x$ on the first $m$ coordinates, and when the verifier is given the
  Grigoriev $\txor$ instance applied only to the first $m$ coordinates of $x'$. The
  soundness properties in \lemref{gnn-soundness} imply that
  $\max_{x} f'(x) \leq 1 - \Omega(1/n^2) := 1 - \delta(n)$. At the same time, the
  Grigoriev pseudosolution also yields a pseudosolution for which $\psE_x[f'(x)] =
  1$.
  Hence, $\mathrm{deg_{SoS}}(1 - f'(x)) \geq \Omega(m)$, and thus by \lemref{LRS}, the
  matrix
  \[ M_n^{f'}(S,x) = 1 - \eps(n) - f'(x_S) \]
  has PSD rank at least $n^{\Omega(m)}$. Now, \emph{this} matrix is
  indeed an exact submatrix of $M_{c,s}(\Phi_n,x)$ as defined above for
  $c = 1 - \eps(n), s = 1  - \delta(n)$. So any \SDPEF{} to
  $\qmat$ has size at least $n^{\Omega(m)} = n^{\log n/
    \poly\log\log n}$.

  Now, note that the trivial reduction from $\qmat$ to $h_{\Sep}$ is
  an \emph{embedding} reduction. Moreover, we claim that it is
  $(1 - \delta(n), 1 - \delta(n))$-approximate. To see this, note that whenever $\qmat \leq 1-\delta(n)$, the
  underlying $\txor$ instance must be infeasible, and thus $h_{\Sep}
  \leq 1 - \delta(n)$ by \lemref{gnn-soundness}. Finally, it remains
  to translate our result into a bound in terms of the dimension of
  the state $d$. Recall that the
  dimension $d$ of the $h_{\Sep}$ instance is polynomially related to
  the number of variables $n$, since the witness states consist of
  $O(\log n)$ qubits. Hence, by the previous
  result for $\qmat$ and \propref{sdp_extend}, we conclude that for an
  appropriate $\delta'(d) = O(1/d^2)$, any
  $(1 - \eps'(d), 1-\delta'(d))$-approximate \SDPEF{} to
  $h_{\Sep(d,d)}$ with $\eps'(d) < \delta'(d)$ must
  have size at least $d^{\log d/ \polylog\log d}$.
\end{proof}
Using the equivalence between $2\ra 4$ norm and $h_{\Sep}$, we can
likewise obtain an SDP hardness result for the $2\ra 4$ norm.
\begin{corollary}\label{cor:24lrs}
  Any \SDPEF{} to $\|A\|_{2\ra 4}$ for $d$-dimensional tensors achieving a multiplicative
  error of $C = 1/O(d^2)$ must have size at least  $\Omega(d^{\log(d)/\polylog \log(d)})$.
\end{corollary}

\subsection{The no-disentangler conjecture} \label{sec:hsep_lb_app}
One application of our result is to prove a version of the Approximate Disentangler
Conjecture for a particular range of parameters. This conjecture was originally formulated by Watrous and first published
in~\cite{AaronsonBDFS08}.  Previously the only evidence in favor of this conjecture was based on complexity assumptions (e.g.~the ETH) and even those results did not rule out the possibility of disentangling maps that were hard to compute.

\begin{definition}
  Let $\cH$ and $\cK$ be Hilbert spaces, and denote the space of
  density matrices on $\cH$ by $D(\cH)$ (likewise $D(\cK)$). A linear
  CPTP map $\Lambda: D(\mathcal{H}) \to D(\mathcal{K} \otimes \mathcal{K})$
  is an \emph{$(\eps, \delta)$-approximate disentangler} if 
  \begin{itemize}
  \item For every $\rho \in D(\cH)$, $\Lambda(\rho)$ is $\eps$-close
    in trace distance to a separable state in $\Sep(\cK \otimes \cK)$.
  \item For every separable state $\sigma \in \Sep(\cK \otimes \cK)$,
    there exists a $\rho \in D(\cH)$ such that $\Lambda(\rho)$ is
    $\delta$-close in trace distance to $\sigma$.
  \end{itemize}
\end{definition}

Our result is the following:
\begin{theorem}\label{thm:disentangler}
  Let
  $d = \mathrm{dim}(\cK)$, and suppose that $\Lambda: D(\cH) \to D(\cK \otimes \cK)$ is an $(\eps,
  \delta)$-approximate disentangler with $\eps + \delta < 1/\poly(d)$. Then
  \[ \mathrm{dim}(\cH) \geq
 \Omega(d^{\log(d) / \polylog\log(d)}). \]
\end{theorem}
This is considerably weaker than Watrous's original formulation, which had $\eps + \delta < O(1)$ and
$\mathrm{dim}(\cH) \geq \exp(d)$. Conditional on ETH, the result was known with $\eps + \delta < O(1)$ and $\mathrm{dim}(\cH) \geq d^{\log(d)/ \polylog\log(d)}$; we are unable to match this due to technical limitations of the LRS result, which are described in more detail in the previous subsection.

\begin{proof}
  We show that a disentangler can be used as an \SDPEF{} to $h_{\Sep}$, thus allowing us to apply Theorem~\ref{thm:lrs-gnn}. Throughout the proof, let $d \equiv \mathrm{dim}(\cK)$. First,
  let us consider the $\delta = 0$ case. We define the optimization
  problem
  \[ h_{\Lambda}(M) \equiv \max_{\rho \in D(\cH)} \Tr[M
  \Lambda(\rho)]. \]
  Note that $h_{\Lambda}(M)$ is a semidefinite program with size
  $\mathrm{dim}(\cH)$. Moreover, when $\delta = 0$, we claim that
  there exists an embedding reduction $R_\Lambda$ from $h_{Sep}$ to
  $h_{\Lambda}(M)$, that achieves a $(s+ \eps, s)$
  approximation. This reduction simply maps an instance $M$ of
  $h_{\Sep}$ to the instance $h_{\Lambda}(M)$ given by the same
  measurement operator $M$. The embedding property follows from the
  definitions: for every separable $\sigma$, there exists a $\rho$
  such that $\Lambda(\rho) = \sigma$, and so $\Tr[M \Lambda(\rho)] =
  \Tr[M \sigma]$. Similarly, the soundness of the reduction also
  follows from the definition: for every $\rho$, $\Lambda(\rho)$ is
  $\eps$-close to some separable $\sigma$. This means that
  $\max_{\rho} \Tr[M \Lambda(\rho)] \leq \max_{\sigma \in \Sep} \Tr[M
  \sigma] + \eps$. Consequently, $h_{\Lambda}(M)$ is a \SDPEF{} of $h_{Sep}$. So, by applying theorem~\ref{thm:lrs-gnn},
  we conclude that $\mathrm{dim}(\cH) \geq d^{\Omega(log(d) / \polylog\log(d))}$.

  Now, let us consider the general case, where $\delta > 0$. In this
  case, we cannot directly apply the preceding argument, since there
  is no embedding from $h_{\Sep}$ to $h_{Disentangled}$. We will fix
  this by using the following gadget: Let $\cB_\delta$ be the set of
  states in $D(\cK \otimes \cK)$ of trace norm less than or equal to
  $\delta$. Then, given an $(\eps, \delta)$-disentangler $\Lambda$, we
  define a new map $\tilde{\Lambda}: D(\cH) \oplus \cB_\delta \to
  D(\cK \otimes \cK)$ by
  \[ \tilde{\Lambda}(\rho \oplus \sigma) \equiv \Lambda(\rho) +
  \sigma. \]
  We claim that for every separable $\tau \in \Sep(\cK, \cK)$,
  there exists a preimage $\rho \in D(\cH), \sigma \in \cB_\delta$
  with $\tilde{\Lambda}(\rho \oplus \sigma) = \tau$. Indeed, since $\Lambda$
  is an $(\eps, \delta)$-disentangler, we know that
  $\tau$ had an \emph{approximate} preimage $\rho$ satisfying
  $\Lambda(\rho) = \tau + \sigma$ for $\| \sigma \|_1 \leq
  \delta$. From our definition of $\tilde{\Lambda}$ it follows that
  $\tilde{\Lambda}(\rho \oplus \sigma) = \tau$ as desired. We also claim
  that for every $\rho \in D(\cH), \sigma \in \cB_\delta$,
  $\tilde{\Lambda}(\rho \oplus \sigma)$ is within $\eps + \delta$ in
  trace distance of some separable state. To see this, note that
  $\Lambda(\rho)$ is within $\eps$ distance of some separable state $\tau$,
  and since $\|\sigma \|_1 \leq \delta$, adding $\sigma$ can increase
  the distance to $\tau$ by at most $\sigma$.

These two claims tell us
  that $\tilde{\Lambda}$ is  ``almost'' an $(\eps + \delta,
  0)$-approximate disentangler: the only catch is that it is not a
  CPTP map acting on quantum states. Nevertheless, we can still
 use the same argument as the $\delta = 0$ case. We define the optimization
  problem $h_{\tilde{\Lambda}}(M)$ by the SDP
  \[ \begin{aligned} 
    &\max_{\rho, \sigma^+, \sigma^-} && \Tr[M \tilde{\Lambda}(\rho
    \oplus \sigma) ] \\
    &\text{such that} && \Tr[\rho] = 1 \\
    &&& \Tr[\sigma^+ + \sigma^-] \leq \delta \\
    &&& \rho, \sigma^+, \sigma^- \succeq 0.
    \end{aligned} \]
  This SDP implements the constraint $\| \sigma \|_1  \leq \delta$. As
  before, consider the reduction from $h_{\Sep}$ to
  $h_{\tilde{\Lambda}}$ that maps the instance $M$ to the instance
  corresponding to the same measurement operator $M$. We claim that
  this reduction is an embedding and is $(s, s+ \eps +
  \delta)$ approximate for any $s$; these claims are proved by a
  similar argument to the $\delta = 0$ case. Thus, we have shown that
  $h_{\tilde{\Lambda}}$ is a $(s, s + \eps + \delta)$-approximate \SDPEF{} of $h_{\Sep}$. So once again, applying
  Theorem~\ref{thm:lrs-gnn} tells us that the dimension of the SDP
  $h_{\tilde{\Lambda}}$ must be at least $\Omega(d^{\log d /
    \polylog\log d})$. Now, the dimension of $h_{\tilde{\Lambda}}$ is
  equal to $\mathrm{dim}(\cH) + d$, so all together we get
  \[ \mathrm{dim}(\cH) \geq \Omega(d^{\log d /
    \polylog\log d}) - d \geq \Omega(d^{\log d /
    \polylog\log d}). \]
\end{proof}

\section{Lower Bounds for Entangled Games} \label{sec:npa}
In this section, we show lower bounds on \SDPEF{}s for the
entangled value of quantum games. First, we review some basic notions.
\begin{definition}
  A \emph{nonlocal game} $G = (Q, A, \pi, V)$ is a game played between
  a referee, or ``verifier,'' and two players, or ``provers.'' In one
  round, the verifier chooses two random questions $q_1, q_2 \in Q$
  according to the joint probability distribution $\pi(q_1, q_2)$, and
  sends $q_1$ to player $1$ and $q_2$ to player $2$. Each player
  returns an answer in the set $A$. The verifier then \emph{accepts}
  the provers' answers with probability given by $V(q_1, q_2, a_1,
  a_2)$. The \emph{winning probability} of a strategy is the
  probability that the verifier accepts when the players play
  according to the strategy.
\end{definition}
Strategies can be either classical or entangled.
\begin{definition}
  A (deterministic) \emph{classical strategy} for a nonlocal game consists of
  functions $f_1, f_2: Q \to A$, with $f_1(q)$ being the answer that
  player 1 gives to question $q$, and likewise for $f_2(q)$. The
  \emph{classical value} $\omega_{classical}(G)$ of a game $G$ is
  the maximum winning probability of a classical strategy for $G$.
\end{definition}
Equivalently, we could have allowed the two players to share classical
random bits. However, by a simple convexity argument one can show that
the classical value of a game is always achieved by a deterministic strategy.
\begin{definition}
  A \emph{quantum strategy} for a nonlocal game consists of:
  \begin{itemize}
  \item Hilbert spaces $H_1, H_2$ and a joint state $\ket{\psi} \in
    H_1 \otimes H_2$,
  \item for every question $q \in Q$, a POVM $\{A^a_q \otimes I_2\}_{a
      \in A}$ acting nontrivially on only player 1's Hilbert space, and
  \item for every question $q \in Q$, a POVM $\{I_1 \otimes B^a_q
    \}_{a \in A}$ acting nontrivially on only player 2's Hilbert space.
  \end{itemize}
  To play the game, each player measures their shared state
  $\ket{\psi}$ using the POVM associated with the question received,
  and returns the POVM outcome as the answer. The \emph{entangled value} $\omega_{\text{entangled}}(G)$ of a game $G$ is
  the maximum winning probability of an entangled strategy for $G$.
\end{definition}

In this section, we show a lower bound on the size of an SDP to
compute the entangled value of a 2-player entangled game, to within
inverse polynomial accuracy. We show both a bound on the size of
general \SDPEF{}s, as well as an explicit integrality gap for the 
non-commuting SoS hierarchy. We do this by embedding 3XOR into a
quantum entangled game, using a result of Ito, Kobayashi, and
Matsumoto~\cite{IKM09}. 
\[
  \txor  \underset{R_1}{\Longrightarrow} \omega_{\textsc{honest classical}} \underset{R_2}{\Longrightarrow} \omega_{\text{entangled}}
\]
The intermediate problem is
\begin{itemize}
\item The problem $\omega_{\textsc{honest classical}}(G)$ is a boolean
  polynomial optimization problem. Each instance is parametrized by a
  2-player game $G$ of the form considered by~\cite{IKM09}. The objective
  function $f(x)$ in the optimization evaluates the winning probability in $G$ of a classical
  strategy parametrized by a boolean string $x$.
\end{itemize}
Before we explain these reductions in more detail, we first review the
result of~\cite{IKM09} that we will use.
\begin{lemma}[Lemma 8 of~\cite{IKM09}]
  Let $\Phi$ be a 3-CSP over $n$
  variables with $m$ clauses\footnote{This is
    called a ``nonadaptive 3-query PCP'' in their language.}. Then there exists a 2-player quantum game $G_\Phi$ such that for some constant $\gamma>0$,
\begin{align*}
  \MXST(\Phi) \leq \omega_{classical}(G_\Phi) &\leq 1 - \frac{1 -
  \MXST(\Phi)}{3} \\
  \omega_{\text{entangled}}(G_\Phi) &\leq 1 - \frac{\gamma(1 - \MXST(\Phi))^2}{m^2}.
\end{align*}
\label{lem:ikm}
\end{lemma}
The game $G_\Phi$ is constructed starting from $\Phi$,
using the technique of \emph{oracularization with a dummy question}.
\begin{definition}
  Let $\Phi$ be a 3-CSP. Then the oracularization of $\Phi$ is a
  2-player entangled game $G_\Phi$. In this game two random clauses
  are sampled from $\Phi$, say acting on bits $(i_1,i_2,i_3)$ and
  $(i_1',i_2',i_3')$, which we assume are randomly ordered.  Then one
  player receives $(i_1,i_2,i_3)$ and the
  other player receives with equal probability either $(i_j,i_1')$ or $(i_1',i_j)$ with $j$ drawn
  randomly from $\{1,2,3\}$.    The players answer with 3 and 2 bits respectively.
  The verifier then accepts if both of the following two checks pass:
  \begin{enumerate}
  \item {\bf Simulation check:} The verifier checks that the answers
    from player 1 satisfy the
    clause associated with variables $i_1, i_2, i_3$.
  \item {\bf Consistency check:} For the variable $i_j$, the verifier
    checks that both players' answers for this variable agree.
  \end{enumerate}
\end{definition}
In our application of this result, we will take the 3-CSP to be
$\txor$. We say that the players are playing \emph{honestly according
  to assignment $x$} if the players responds with the answers
$(x_{i_1}, x_{i_2}, x_{i_3})$ and either $(x_{i_j},x_{i'})$ or
$(x_{i'},x_{i_j})$ as appropriate.
Thus the consistency check will always pass and the simulation check
will pass with probability $f_\Phi(x)$, which is defined to equal the
fraction of clauses in $\Phi$ satisfied by $x$.
\begin{definition}
  The problem $\omega_{\textsc{honest classical}}(G_{\Phi})$ is the
  optimization problem of maximizing $f_\Phi(x)$ over $x \in \zo^n$.
\end{definition}

\begin{lemma}
  Let $\Phi$ be a $\txor$ instance produced
  by Proposition~\ref{prop:grigoriev}. Then there exists a degree-$\Omega(n)$ value-$(1,\frac{1}{2}+\eps)$ integrality gap for $\omega_{\textsc{honest classical}}(G_\Phi)$ for all
  $\eps$.
  \label{lem:sos_lb_games}
\end{lemma}
\begin{proof}
Essentially, this follows directly from the fact that $f_{\Phi}(x)$
counts the
fraction of clauses satisfied by $x$, since the other tests in the
game all pass with probability $1$ for honest strategies.


  In more detail, the function $f_\Phi(x)$ is a polynomial function of the variables
  $x$. Each term in the polynomial corresponds to a possible check
  that the verifier performs. We show that each term has
  pseudoexpectation $1$ under the pseudoexpectation operator $\psE$
  produced by proposition~\ref{prop:grigoriev}. Recall that this pseudoexpectation
  operator has degree $\Omega(n)$.
\begin{itemize}
\item {\bf Simulation test}: In this test, we verify that the
    answers of each prover satisfy the clause they were asked. In other words, for a
    3XOR clause $x_i x_j x_k = b$, we want to verify that the player's
    answers multiply together to $b$. For every clause $b = x_i x_j
    x_k$, we have a term 
    \[ \frac{1}{2} + \frac{1}{2} b x_i x_j x_k , \]
    in the polynomial $f_\Phi(x)$. We compute
    the pseudoexpectation of this term:
    \begin{align*}
      \psE'[ V_{b,A} ] 
      &= \frac{1}{2} + \frac{1}{2} b\psE[x_i x_j x_k ] \\
      &= 1.
    \end{align*}
  \item {\bf Consistency test}: In this test, we check that the two
    players give the same answer when asked about the same bit. This
    test is automatically satisfied for \emph{any} input to
    $f_\Phi(x)$, since any honest strategy is consistent. Thus, it is also satisfied by the pseudoexpectation $\psE$.
  \end{itemize}
  Thus, we have shown that there exists a degree-$\Omega(n)$
  pseudoexpectation $\psE$ such that $\psE[f_\Phi(x)] = 1$. However,
  notice that $f_\Phi(x) = \alpha \Phi(x) + \beta$, where $\alpha$ is
  the probability of doing a simulation test and $\beta$ the
  probability of doing a consistency test. Thus, since
  $\M3XR(\Phi) \leq \frac{1}{2} + \delta$, we deduce that $\max_x
  f_\Phi(x) \leq \alpha(\frac{1}{2} + \delta) + \beta = 1 -
  \alpha(\frac{1}{2} - \delta) \leq 1 -c $ for the appropriate constant $c$. Thus, $\psE$ gives us the desired degree $\Omega(n)$, value-$(1-\eps,1-c)$ integrality gap for all $\eps >0$.
\end{proof}

\subsection{General SDPs}\label{subsec:games-general}
First, we show an SDP lower bound for the optimization over honest strategies.
\begin{lemma}\label{lem:honest-game-lrs}
  Suppose $S_n$ is a sequence of \SDPEF{}s to the problem
  $\omega_{\textsc{honest classical}}(G)$ of size $r_n$, achieving an $(c=1 - \eps(n), s=1 -
  \delta)$-approximation, where $\delta < \frac{1}{2}$ and $\eps(n)
  < \delta$. Then $r_n \geq \Omega\left( n^{\log n/ \poly\log\log n}\right)$.
  \end{lemma}
\begin{proof}
The proof of this theorem is similar to that of~\thmref{lrs-gnn}, and also relies on the result of LRS. A notable simplification that occurs in the games setting is that it is \emph{easier} to embed an instance of a problem into an instance of the same problem with more variables.

Let $f_\Phi(x)$ be the objective function of $\omega_{\textsc{honest classical}}(G_\Phi)$ on an instance of $G$ given by a Grigoriev 3XOR instance $\Phi$ on $m$ variables. By~\lemref{sos_lb_games}, we know that $\mathrm{deg_{SoS}}(1 - \delta - f_\Phi(x)) \geq \Omega(m)$. Thus, by the LRS theorem, for $n = m^{\Omega(m)}$, the matrix $M_n^f(S,x) = 1 - \delta - f_\Phi(x_S)$ has PSD rank at least $n^{\Omega(m)}$. This matrix is a submatrix of $M_{c,s}(\mathcal{J}, x) = 1 - \delta - \langle \mathcal{J}, x \rangle$ which is the pattern matrix of the \SDPEF{}. So the \SDPEF{} has size at least $n^{\Omega(m)} = n^{\log n / \poly \log \log n}$.
\end{proof}

For a general 2-player game $G_N $ of size $N$ (i.e. for which the distribution
$\pi$ over questions and predicate $V$ can be specified in $N$ bits), define the optimization problem
$\omega_{\text{entangled}}(G_N)$ to be the optimization of the winning
probability of game $G_N$ over all two-player entangled strategies. 

\begin{theorem}\label{thm:game-nogo}
  Suppose $S_N$ is a sequence of \SDPEF{}s to the problem
  $\omega_{\text{entangled}}(G_N)$ of size $r_N$, achieving an $(c=1 - \eps(N), s=1 -
  \delta(N))$-approximation, where $\delta(n) = O(1/N)$ and $\eps(N)
  < \delta(N)$. Then $r_N \geq \Omega\left( N^{\log N/ \poly\log\log N}\right)$.
\end{theorem}
\begin{proof}
  As before, let $G_\Phi$ be the oracularized game associated with the
  $\txor$ instance $\Phi$ on $n$ variables, and $f_\Phi(x)$ be the winning probability
  of an honest strategy played according to assignment $x$. From the
  definitions, it follows that there is an embedding reduction from the problem
  $\omega_{\textsc{honest classical}}(G)$ to $\omega_{\text{entangled}}(G_N)$
  where $N = \Theta(n^2)$. Moreover, by
  lemma~\ref{lem:ikm}, this reduction is $(1 - c/n^2, 1 -
  c')$-approximate for constants $c, c'$. Thus, any \SDPEF{} of
  size $r_N$ for
  $\omega_{\text{entangled}}(G_N)$ that achieves a $(1 - \eps(N), 1 - c/N)$
  approximation implies an \SDPEF{} of size $r_N$ for $\omega_{\textsc{honest classical}}(G)$
  that achieves a $(1 - \eps(N), 1 - c')$ approximation. Now, by Lemma~\ref{lem:honest-game-lrs}, any such \SDPEF{} must have size at least $r_N \geq \Omega \left(N^{\log N / \log \log N}\right)$.
\end{proof}

\subsection{An explicit lower bound for ncSoS}\label{subsec:npa-explicit}
In the previous section we gave a lower bound on the size of 
\SDPEF{}s for the problem of computing the entangled game value. We will now present an explicit lower
bound for a family of \SDPEF{}s called the \emph{non-commuting sum of squares} (ncSoS)
hierarchy, also referred to as the NPA hierarchy~\cite{NPA08,DLTW08}. Recall that in the
sum-of-squares hierarchy, one optimizes a polynomial function $f(x)$
by optimizing $\psE_{x \sim \mu} f(x)$ over pseudodistributions $\mu$
that obey certain constraints. Likewise, in the ncSoS hierarchy, the
winning probability $\omega$ is viewed as a polynomial in
non-commuting variables (corresponding to the quantum operators in the
provers' strategy), and the game value is found by optimizing the 
 \emph{non-commuting pseudoexpectation} of this nc polynomial.  A
 non-commuting pseudoexpectation satisfies conditions similar to an
 ordinary pseudoexpectation operator:
\begin{definition}
  An degree-$d$ ncSoS pseudoexpectation is a linear map $\psE[\cdot]$ that maps
  nc polynomials in the provers' measurement operators $\{A^{a}_q\}$,
  $\{B^{a}_q\}$ to real numbers. This map satisfies the following
  properties:
  \begin{itemize}
    \item Normalization: $\psE[I] = 1$.
    \item Positivity: for all polynomials $p(A, B)$ with degree at
      most $d/2$, $\psE[p^\dagger p] \geq 0$.
    \item Commutation: for any operators $A, B$ acting on different 
      provers, 
      $\psE[q_1(x)(AB - BA)q_2(x)] = 0$ for all polynomials $q_1(x),q_2(x)$ with $\deg q_1 + \deg q_2 \leq d - 2$.
  \end{itemize}
\end{definition}
In the following theorem, we show that when the degree $d$ is small
enough, we can construct a non-commuting pseudoexpectation according
to which every test in the game is satisfied with probability $1$,
even though the game value is less than $1 - 1/\poly(n)$. We do this in a quite
generic way, by
constructing a commuting pseudoexpectation for the CSP instance underlying the
game, and ``lifting'' this pseudoexpectation to a non-commuting
pseudoexpectation for the game.
\begin{theorem}\label{thm:ncsos-nogo}
  For every $n$ there exists a two-player entangled game $G$ with
  $O(n)$ questions and three-bit answers, such that
  $\omega_{\text{entangled}}(G) \leq 1 - c/n^2$ for some constant $c$, but there exists a pseudoexpectation of degree $Omega(n)$ according to
  which the game value is $1$.
\end{theorem}
\begin{proof}
  Start with a 3XOR instance with maximum satisfiable fraction $1/2 +
  \epsilon$. Then \lemref{ikm} gives us the first part of the conclusion. For the second part, we explicitly
  construct the pseudodistribution using the Grigoriev
  instance. 
  Let the two players be denoted A and B. 
  Their strategies
  are given by POVMs $\{A_{i_1 i_2 i_3}^{a_1 a_2 a_3}\}$, $\{B_{j_1
    j_2}^{b_1 b_2 }\}$, where $i_1, \dots, i_3, j_1, \dots, j_3$ are
  indices of variables in the $\txor$ instance. To specify a
  pseudodistribution, we need to assign values to every
  pseudoexpectation of words built out of these variables. We do so as
  follows: first, we impose the condition that the $A$ and $B$
  operators are mutually commuting, and moreover that $A_{i_1 i_2
    i_3}^{a_1 a_2 a_3} =  C_{i_1}^{a_1}
  C_{i_2}^{a_2} C_{i_3}^{a_3}, B_{j_1 j_2}^{b_1 b_2} = C_{j_1}^{b_1} C_{j_2}^{b_2}$ where the operators $\{C_{i}^0,
  C_{i}^1\}$ form a projective measurement for every index $i$. For
  convenience, we will henceforth work with the observables $C_{i}
  \equiv C_{i}^0 - C_{i}^1$; these square to identity. Now, let $\psE$
  be the Grigoriev pseudoexpectation operator for the 3XOR
  instance. We define an ncSoS pseudoexpectation $\psE'$ as follows:
  \[ \psE'[C_{i_1} \dots C_{i_k}] \equiv \psE[x_{i_1} \dots
  x_{i_k}]. \]
  By construction, this pseudoexpectation satisfies all the ncSoS
  constraints. It is defined up to
  degree $\Omega(n)$. We now need to check that it achieves a game
  value of $1$. The game consists of two kinds of checks:
  simulation and consistency. 
  \begin{itemize}
  \item {\bf Simulation test}: In this test, we verify that the
    answers of each prover satisfy the clause they were asked. In other words, for a
    3XOR clause $x_i x_j x_k = b$, we want to verify that player $A$'s
    answers multiply together to $b$. For every clause $b = x_i x_j
    x_k$, we have a term 
    \[ V_{b, A} =\left(  \sum_{a_i, a_j, a_k} a_i a_j a_k A_{x_i x_j
      x_k}^{a_i a_j a_k} \right) \otimes I_B, \]
    in the game value, and an analogous term for player $B$. We compute
    the pseudoexpectation of this term:
    \begin{align*}
      \psE'[ V_{b,A} ] &= \sum_{a_i, a_j, a_k} (a_i a_j a_k b)
                         \psE'[ A_{x_i x_j
      x_k}^{a_i a_j a_k} \otimes I_B] \\
      &= \sum_{a_i a_j a_k} a_i a_j a_k b \psE' [ C_{x_i}^{a_i}
        C_{x_j}^{a_j} C_{x_k}^{a_k} \otimes I_B] \\
      &= b \psE' [ C_{x_i} C_{x_j} C_{x_k} \otimes I_B] \\
      &= b\psE[x_i x_j x_k ] \\
      &= 1.
    \end{align*}
  \item {\bf Consistency test}: In this test, we check that players A
    and B give the same answer if asked about the same bit.
    \begin{align*}
      V_{i,A,B} &= \E_{j, k, p} \sum_{a_1 a_2 a_3} \sum_{b_1 b_2} (a_1 b_1) A^{a_1 a_2 a_3}_{x_i x_j x_k} \otimes B^{b_1 b_2}_{x_i x_p }\\
      \psE'[V_{i,A,B}] &= \E_{j, k, p} \sum_{a_1 a_2 a_3} \sum_{b_1 b_2} (a_1 b_1) \psE'[A^{a_1 a_2 a_3}_{x_i x_j x_k} \otimes B^{b_1 b_2}_{x_i x_p}] \\
                &= \E_{j, k, p} \sum_{a_1 a_2 a_3} \sum_{b_1 b_2} (a_1 b_1) \psE'[C_i^{a_1} C_j^{a_2} C_k^{a_3} C_i^{b_1}
                  C_p^{b_2} ] \\
                &= \sum_{a_1 b_1} (a_1 b_1) \psE'[C_i^{a_1} C_i^{b_1}]
      \\
                &= \psE'[C_i C_i] \\
                &= \psE[x_i^2] \\
                &= 1.
    \end{align*}
  \end{itemize}
\end{proof}

\section*{Acknowledgement} 
AWH and AN were funded by NSF grant CCF-1629809 and AWH was funded by NSF grant
CCF-1452616. XW was funded by the NSF Waterman Award of Scott Aaronson.  All three authors
(AWH, AN and XW) were funded by ARO contract W911NF-12-1-0486.  We are grateful to an
anonymous STOC reviewer for pointing out a mistake in an earlier version.

\newcommand{\etalchar}[1]{$^{#1}$}

\end{document}